\documentclass[journal,12pt,draftclsnofoot,onecolumn]{IEEEtran}

\usepackage{cite}

\usepackage{subfigure}
\usepackage{graphicx}

\usepackage{epstopdf}

\usepackage{amsmath}
\usepackage{amssymb}
\usepackage{amsfonts}
\usepackage{amsthm}

\usepackage{algorithmic}
\usepackage{algorithm}

\usepackage{array}
\usepackage{multirow}
\usepackage{multicol}

\usepackage{setspace}
\usepackage{float}

\usepackage{url}

\usepackage{lineno}
\usepackage{lipsum}
\usepackage{latexsym,bm}

\theoremstyle{definition}
\newtheorem{lemma}{Lemma}
\newtheorem{proposition}{Proposition}
\newtheorem{corollary}{Corollary}

\theoremstyle{remark}

\begin{document}

\title{On the Energy-Efficiency Trade-off Between \\ Active and Passive Communications with \\ RIS-based Symbiotic Radio}

\author{Sihan~Wang, Jingran~Xu, and Yong~Zeng,~\IEEEmembership{Senior Member,~IEEE,} 
\thanks{This work was supported by the National Key R$\text{\&}$D Program of China with Grant number 2019YFB1803400.
Part of this work has been presented at the 2022 IEEE WCSP, Nanjing, China, 01-03 Nov. 2022\cite{10039158}.}
\thanks{ The authors are with the National Mobile Communications Research Laboratory, 
Southeast University, Nanjing 210096, China. Y. Zeng is also with the Purple Mountain Laboratories, Nanjing 211111, China (Email: \{turquoise, jingran\_xu, yong\_zeng\}@seu.edu.cn). \emph{(Corresponding author: Yong Zeng.)}}}

\maketitle
\begin{abstract}
  Symbiotic radio (SR) is a promising technology of spectrum- and energy-efficient wireless systems, for which the key idea is to use cognitive backscattering communication to achieve mutualistic spectrum and energy sharing with passive backscatter devices (BDs). In this paper, a reconfigurable intelligent surface (RIS) based SR system is considered, where the RIS is used not only to assist the primary active communication, but also for passive communication to transmit its own information. For the considered system, we investigate the EE trade-off between active and passive communications, by characterizing the EE region. To gain some insights, we first derive the maximum achievable individual EEs of the primary transmitter (PT) and RIS, respectively, and then analyze the asymptotic performance by exploiting the channel hardening effect. To characterize the non-trivial EE trade-off, we formulate an optimization problem to find the Pareto boundary of the EE region by jointly optimizing the transmit beamforming, power allocation and the passive beamforming of RIS. The formulated problem is non-convex, and an efficient algorithm is proposed by decomposing it into a series of subproblems by using alternating optimization (AO) and successive convex approximation (SCA) techniques. Finally, simulation results are presented to validate the effectiveness of the proposed algorithm.
\end{abstract}


%
\IEEEpeerreviewmaketitle

\section{Introduction}\label{s1}

The sixth generation (6G) mobile communication networks are expected to support ultra-broadband transmission, ultra-massive access, ubiquitous sensing, and reliable intelligent connectivity\cite{dlr133477,You2020Towards6W,8820755}. However, with the dramatic increase of connected communication devices, there are two serious challenges: the shortage of spectrum resources and the sustainability of energy supply. To resolve such issues, it is imperative to develop innovative technologies to simultaneously improve the spectrum efficiency (SE) and energy efficiency (EE) to fully realize the vision of Internet of Everything (IoE).

One promising solution to address the above challenge is symbiotic radio (SR) communication\cite{8907447}, which combines the advantages and effectively avoids the deficiencies of both cognitive radio (CR)\cite{5783948} and cooperative ambient backscatter communication (AmBC)\cite{8692391}. The key idea of SR is to leverage cognitive backscattering communication to achieve mutualistic spectrum and energy sharing by integrating passive backscatter device (BD) with active primary transmitter (PT). Specifically, the BD modulates its own information by passively backscattering the incident signal from the PT without active signal processing. As such, apart from enhancing SE as in the conventional CR system, SR exploits the passive backscattering technology to greatly reduce the power consumption as in the AmBC system, which is expected to improve EE significantly\cite{9193946}.
In general, SR can be classified into two categories, parasitic SR (PSR) and commensal SR (CSR), according to the relationships between the symbol periods for the BD and the PT\cite{8907447}. In PSR, the signals of both BD and PT have equal symbol durations, so that the backscattering transmission and the primary transmission interfere with each other. By contrast, in CSR, the symbol duration of the passive BD signal spans multiple PT symbol durations, which may contribute additional multipath signal components to enhance the active primary transmission\cite{8907447}. SR communication is expected to find a wide range of applications, such as E-health, wearable devices, environmental monitoring, vehicle-to-everything (V2E), and smart city\cite{2021arXiv211108948B}.

Significant research efforts have been recently devoted to the study of SR communications, e.g., in terms of theoretical analysis\cite{8638762} and performance optimization to maximize the achievable rate\cite{8665892}, channel capacity\cite{9600844} and EE\cite{9036977}. However, a practical challenge of SR technique is that due to the double-hop signal attenuations, the backscattering link is typically much weaker than the direct primary link\cite{9193946}. Thus, the performance of the secondary passive communication and its additional multipath contribution to the primary communication link are quite limited.
To address such issues, various techniques have been proposed to enhance the backscattering links, such as massive BDs enabled SR\cite{9686018,Xu2022MIMOSR,9461158} or active-load assisted SR\cite{9154299} communications. In particular, multiple-input multiple-output (MIMO) SR communication system with massive number of BDs is studied in \cite{Xu2022MIMOSR}, where closed-form expressions of the asymptotic regime are derived to reveal the relationship between the primary and secondary communication rates. Besides, a precoding optimization problem is solved to maximize the primary communication rate while guaranteeing the minimum secondary communication rate.

On the other hand, reconfigurable intelligent surface (RIS), also termed as intelligent reflecting surface (IRS)\cite{9122596}, has emerged as another promising solution to strengthen the backscattering link. RIS is composed of a large number of passive reflecting elements, which is able to configure the wireless environment in a desirable manner by adjusting the reflection coefficients without relying on active radio frequency (RF) chain components\cite{9140329,9086766}. This prominent property renders RIS rather appealing to enhance the performance of various wireless communication systems, such as CR\cite{9235486}, MIMO\cite{9110912}, unmanned aerial vehicle\cite{9351782} and non-orthogonal multiple access (NOMA) systems\cite{9133094}.
Unlike such works on RIS-assisted communications, for RIS-based SR systems, the RIS is used not only as a helper to assist the primary active communication, but also as a BD to enable secondary passive communication to transmit its own information\cite{9537929}. To distinguish it with the conventional counterpart, we term such kind of RIS as RIS-BD in this paper.

Performance optimization for RIS-BD based SR systems has been studied for different purposes, such as power minimization\cite{9481926,9345739,9839116} and channel capacity maximization\cite{9530367}.
Specifically, in \cite{9481926}, an algorithm based on generalized power method (GPM) technique is proposed to minimize the transmit power in a RIS empowered SR over broadcasting
signals. A RIS-assisted MIMO symbiotic communications adopting multiple reflecting patterns is investigated in \cite{9530367} to maximize the capacity. 
Furthermore, novel schemes that incorporate RIS-based SR for symbiotic active/passive communications have also attracted increasing interest. In \cite{9982407}, an optimization framework for the symbiotic operation of a multiuser CR network consisting of a NOMA-based primary network and a RIS-based secondary network is developed. The authors of \cite{9951145} propose RIS-aided number modulation, where the number of RIS elements are divided into the in-phase and quadrature subsets to transmit the RIS's information.

However, it is worth noting that there are only very limited works on EE study for RIS-BD based SR communication systems. A RIS-assisted SR communication network with multiple primary users (PUs) and multiple clusters of IoT devices linked with a RIS is considered in \cite{9817403}. The authors maximize the EE by using alternating optimization (AO) together with semi-definite relaxation (SDR) and Dinkelbach's algorithm. In \cite{10013073}, the authors propose a method based on the accelerated generalized Benders decomposition (GBD) algorithm to maximize the EE of the secondary receiver under a required signal-to-interference-plus-noise ratio (SINR) constraint for the primary receiver (PR). 
It is worth remarking that such existing EE studies of RIS-BD based SR systems mainly focus on the so-called global EE, defined as the ratio of the weighted sum-rate of primary and backscattering communications to the total power consumption of passive and active devices. However, considering global EE may lead to the overlook of the EE of RIS-BD since 
both the communication rate and power consumption of active communication are typically orders of magnitude higher than that of the passive communication.
Different from global EE, studying the individual EEs of passive RIS-BD and active PT may reveal the fundamental relationship of EEs between active and passive communications. Therefore, in this paper, we investigate a RIS-BD based multiple-input single-output (MISO) SR communication system. In order to develop an insightful analysis of the EE trade-off between active primary and passive backscattering communications, EE maximization problem is formulated to characterize the EE region. Our main contributions are summarized as follows:
\begin{itemize}
  \item First, we present the system model of RIS-BD based MISO SR communication systems, and then derive the maximum individual EEs of the active primary communication and the passive backscattering communications. We show that achieving these two maximum EE values require significantly different transmission strategies, which implies that there exist a nontrivial trade-off between these two EEs.
  \item Next, to exploit the channel hardening effect for RIS-BD based SR systems, we provide the asymptotic analysis by assuming that the number of PT antennas or RIS-BD elements goes very large. We analyze the distribution characteristics of these two EEs in general Rician channel. Besides, closed-form expressions are derived for the EEs of active and passive communications under some special channel assumptions to get some insights.
  \item Furthermore, to study the fundamental EE trade-off between active and passive communications, we formulate an optimization problem to characterize the Pareto boundary of the EE region. The formulated problem is challenging to be solved, since the variables are coupled and both the objective and the constraints are nonconvex. We propose an effective algorithm termed as sample-average based bisection approach with AO and successive convex approximation (SCA) techniques to transform it into a series of low-complexity convex subproblem. Simulation results are provided to validate our theoretical analysis.
\end{itemize}

The rest of this paper is organized as follows. Section \ref{s2} presents the system model of RIS-BD based MISO SR communication. Section \ref{s3} derives the maximum individual EEs for the PT and the RIS-BD, respectively. In Section \ref{s4}, asymptotic performance analysis is provided. In Section \ref{s5}, the Pareto boundary of the EE region is characterized. Simulation results are presented in Section \ref{s6}. Finally, Section \ref{s7} concludes this paper.

Notations:
Lower- and uppercase letters $x$ and $X$ denote a scalar (or constant) and random variable, respectively. Boldface lower- and uppercase letters $\mathbf{x}$ and $\mathbf{X}$ denote vector and matrix, respectively. Notations $x^*$ and $|x|$ denote the conjugate and the absolute value of a scalar, respectively. The L1 norm and L2 norm (also called Euclidean norm) of a vector $\mathbf{x}$ are denoted respectively as $\|\mathbf{x}\|_1$ and $\|\mathbf{x}\|$. For a matrix $\mathbf{X}$, denote its conjugate, transpose, and conjugate transpose as ${\mathbf{X}}^* $, $\mathbf{X}^\mathrm{T}$, and $\mathbf{X}^\mathrm{H}$, respectively. $\mathbf{I}_M$ denotes an $M \times M$ identity matrix.
${\mathbb{C}^{M \times N}}$ denotes the space of ${M \times N}$ matrices with complex entries. $\mathcal{CN} (\mu,\Sigma)$ denotes the circularly symmetric complex Gaussian (CSCG) distribution with mean $\mu$ and variance $\Sigma$.
${\text{diag}}\left( \mathbf{x} \right)$ denotes a diagonal matrix whose diagonal elements are given by vector $\mathbf{x}$. ${\mathbb{E}_X}\left[ \cdot \right]$, ${\mathbb{E}[X]}$, $\mathrm{Var} [X]$ and $\mathrm{Cov}[X,Y]$ denote the statistical expectation with respect to $X$, the expectation and variance of $X$, and the covariance between $X$ and $Y$, respectively. Furthermore, notations $\mathrm{W}\left(\cdot \right)$, $\arg \left( \cdot \right)$ and $\mathrm{L}_q (\cdot)$ denote the Lambert-W function, the phase of any complex number and the Laguerre polynomial of order $q$, respectively.

\section{System Model}\label{s2}
\vspace*{-.5\baselineskip}
\begin{figure}[H]
  \centering
  \includegraphics[width = .55\textwidth]{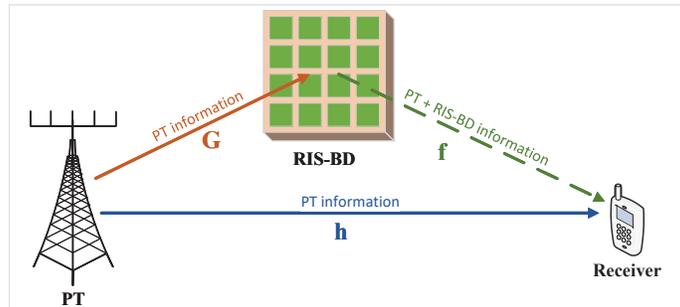}
  \vspace*{-.5\baselineskip}
  \caption{A RIS-BD based SR communication system.}
  \vspace*{-.5\baselineskip}
  \label{fig_1}
\end{figure}
Fig. \ref{fig_1} shows a RIS-BD based MISO SR communication system, which includes an active PT with $M$ antennas, a passive RIS-BD with $N$ reflecting elements and a single-antenna receiver. Both the PT and RIS-BD wish to communicate with the receiver. The PT actively transmits its information-bearing signal to the receiver via multi-antenna beamforming. Meanwhile, the RIS-BD not only assists the primary transmission but also modulates its own information over the incident signal from the PT by cognitive backscattering communication technology. Thus, the RIS-BD reuses not only the spectrum, but also the power of the PT to transmit its own information.

Denote the MISO channel of the direct PT-to-receiver link as ${{\mathbf{h}}} = \left[{{h_1},\cdots,{h_M}} \right] ^\mathrm{T} \in {\mathbb{C}^{M \times 1}}$, where $h_{m}$ is the channel coefficient between the $m$-th antenna of the PT and the receiver. Further, denote the channel matrix between the PT and RIS-BD as ${\mathbf{G}} = \left[ {{g_{nm}}} \right] \in {\mathbb{C}^{N \times M}}$, where ${{g_{nm}}}$ is the channel coefficient between the $m$-th antenna of the PT and the $n$-th reflecting element of the RIS-BD. Moreover, the MISO channel of the RIS-BD-to-receiver link is denoted as ${{\mathbf{f}}} = \left[ {{f_1},{f_2},\cdots,{f_N}} \right] ^ {\mathrm{T}} \in {\mathbb{C}^{N \times 1}}$, where ${{f_{n}}}$ is the channel coefficient between the $n$-th reflecting element of the RIS-BD and the receiver.

Since the RIS-BD operates as a low-power passive device, its own communication rate is much lower than that of the PT\cite{9347977}. Therefore, we consider CSR setup, i.e., $T_c = L {T_s}, L \gg 1$, where $T_c$ and $T_s$ is the symbol duration of the RIS-BD and the PT, respectively.
Let $s\left( l \right)$ be the independent and identically distributed (i.i.d.) information-bearing symbol following the standard CSCG distribution, i.e., $s\left( l \right) \sim \mathcal{C}\mathcal{N}\left( {0,1} \right)$. Further, denote $c \sim \mathcal{CN} (0,1)$ as the information-bearing symbol of the RIS-BD to be transmitted in one RIS-BD symbol period, which spans $L$ symbol periods of the primary signal $s\left(l\right)$, for $l = 1,\cdots, L$. Further, denote the active transmit beamforming vector of the PT as ${\bf{w}}\in{{\mathbb C}^{M \times 1}}$, which satisfies ${\left\| {\mathbf{w}} \right\|^2} \le P_{\max}$, with $P_{\max}$ denoting the maximum allowable transmit power of the PT.

Denote by ${\boldsymbol{\theta}}=\left[{\theta _1},\cdots,{\theta _N}\right]^\mathrm{T} \in {\mathbb{C}^{N \times 1}}$ the phase shift vector of the RIS-BD, where ${\theta _n}\in [0, 2\pi)$ is the phase shift of the $n$-th reflecting element, for $n = 1, ..., N$. Let ${\boldsymbol{\phi }} = \left[ {{e^{j{\theta _1}}},{e^{j{\theta _2}}},\cdots,{e^{j{\theta _N}}}} \right]^{\mathrm{H}} \in {\mathbb{C}^{N \times 1 }}$ denote the reflection coefficient vector.
Therefore, the reflected signal from the RIS-BD can be expressed as $\sqrt \rho  {{\mathbf{f}}^{\mathrm{H}}}{\mathbf{\Phi Gw}}s\left( l \right)c$, where ${\mathbf{\Phi }} = {\text{diag}}\left( {\boldsymbol{\phi }}^{\mathrm{H}} \right) \in {\mathbb{C}^{N \times N}}$ is the reflection coefficient matrix of the RIS-BD, and $0 < \rho \le 1$ denotes the reflection efficiency.

The received signals at the receiver for each backscattering symbol period can be written as
\begin{equation}\label{eq18}
  y\left( l \right) = \; {{\mathbf{h}}^{\mathrm{H}}}{\mathbf{w}}s\left( l \right) + \sqrt \rho  {{\mathbf{f}}^{\mathrm{H}}}{\mathbf{\Phi Gw}}s\left( l \right)c + z\left( l \right) = \big( {{{\mathbf{h}}^{\mathrm{H}}} + \sqrt \rho  {{\mathbf{f}}^{\mathrm{H}}}{\mathbf{\Phi G}}c} \big){\mathbf{w}}s\left( l \right) + z\left( l \right), 
\end{equation}
where $l = 1,\cdots, L$ and $z\left(l\right)\sim\mathcal{CN}\left({0,{\sigma^2}}\right)$ is the additive white Gaussian noise (AWGN) with zero mean and power ${\sigma^2}$.

The second term in (\ref{eq18}) can be viewed as the output of the primary signal $s\left(l\right)$ passing through a channel $\sqrt \rho  {{\mathbf{f}}^{\mathrm{H}}}{\mathbf{\Phi G}}c$ that varies depending on the passive backscattering signal $c$. Based on the received signal (\ref{eq18}), the receiver first decodes the active primary signal $s\left(l\right)$ by treating the passive signal $c$ as a multi-path component, and the equivalent channel for decoding $s\left(l\right)$ is denoted by ${\mathbf{h}}_{\mathrm{eq}}^{\mathrm{H}}(c) = {{\mathbf{h}}^{\mathrm{H}}} + \sqrt \rho  {{\mathbf{f}}^{\mathrm{H}}}{\mathbf{\Phi G}}c$.
Therefore, the signal-to-noise ratio (SNR) for decoding $s\left(l\right)$ at the RIS-BD with a given $c$ is
\begin{equation}\label{eq19}
  {\gamma _s}\left( c \right) = \frac{| {\big( {\mathbf{h}}_{\mathrm{eq}}^{\mathrm{H}}(c) \big){\mathbf{w}}} |^2}{{{\sigma ^2}}} = \frac{| {\big( {{{\mathbf{h}}^{\mathrm{H}}} + \sqrt \rho  {{\mathbf{f}}^{\mathrm{H}}}{\mathbf{\Phi G}}c} \big){\mathbf{w}}} |^2}{{{\sigma ^2}}}.
\end{equation}
From (\ref{eq19}), the expression of ${\gamma _s}\left( c \right)$ is related to the passively modulated signal of the RIS-BD $c$, which changes relatively fast compared to the channel variation\cite{9391685}. Thus, by taking expectation over the random passive signal $c$, the average achievable rate of the primary transmission is
\begin{equation}\label{eq20}
  {R_s} = B{\mathbb{E}_c}\left[ {{{\log }_2}\left( {1 + {\gamma_s}\left( c \right)} \right)}\right] = B{\mathbb{E}_c}\Big[ {{{\log }_2}\big( {1 + \frac{| {( {{{\mathbf{h}}^{\mathrm{H}}} + \sqrt \rho  {{\mathbf{f}}^{\mathrm{H}}}{\mathbf{\Phi G}}c} ){\mathbf{w}}} |^2}{{{\sigma ^2}}}} \big)} \Big].
\end{equation}
After decoding $s\left(l\right)$, the primary signal can be subtracted from (\ref{eq18}) before decoding the passive signal $c$. Specifically, for each RIS-BD symbol period, we denote the primary signal, the noise and the intermediate signal after removing the primary signal in vector by ${\mathbf{s}} = {[s(1),s(2),\cdots,s(L)]^\mathrm{T}}$, ${\mathbf{z}} = {[z(1),z(2),\cdots,z(L)]^\mathrm{T}}$ and $\hat {\mathbf{y}} = {[y(1),y(2),\cdots,y(L)]^\mathrm{T}}$. Then it follows from (\ref{eq18}) that
\begin{equation}\label{eq1}
  \hat {\mathbf{y}} = \sqrt \rho  {{\mathbf{f}}^{\mathrm{H}}}{\mathbf{\Phi Gws}}c + {\mathbf{z}}.
\end{equation}
With ${\mathbf{s}}$ decoded at the receiver, the maximal ratio combining (MRC) can be applied by pre-multiplying $\hat {\mathbf{y}}$ in (\ref{eq1}) by $\frac{1}{{\sqrt L }}{{\mathbf{s}}^{\mathrm{H}}}$. For $L \gg 1$, due to the law of large numbers and the fact that the information-bearing symbols $s(l)$ are i.i.d. random variables with variance 1, we have $\frac{1}{L}{{\mathbf{s}}^{\mathrm{H}}}{\mathbf{s}} \to 1$. Therefore, the resulting signal can be obtained by
\begin{equation}\label{eq34}
  \tilde {{y}} = \frac{1}{{\sqrt L }}{{\mathbf{s}}^{\mathrm{H}}}{\hat {\mathbf{y}}}= \sqrt {L \rho}  {{\mathbf{f}}^{\mathrm{H}}}{\mathbf{\Phi Gw}}c + \tilde {{z}},
\end{equation}
where $\tilde {{z}} = \frac{1}{{\sqrt L }}{{\mathbf{s}}^{\mathrm{H}}}{\mathbf{z}}$ is the resulting noise which can be shown to follow the distribution $\mathcal{CN}( {0,{\sigma ^2}})$.
As a result, the achievable rate of the RIS-BD is
\begin{equation}\label{eq22}
  {R_c} = \frac{B}{L}{\log _2}({1 + \gamma _c} ) = \frac{B}{L}{\log _2}\Big( {1 + \frac{{\rho L{{| {{{\mathbf{f}}^{\mathrm{H}}}{\mathbf{\Phi Gw}}} |}^2}}}{{{\sigma ^2}}}} \Big),
\end{equation}
where the denominator ${L}$ accounts for the fact that in the CSR setup, only one RIS-BD symbol is transmitted during $L$ successive primary symbol periods, and the primary signal can be viewed as a spread-spectrum code with length $L$ for RIS-BD symbols. Therefore, the SNR for decoding RIS-BD symbol ${\gamma _c}$ is increased by $L$ times, at the cost of symbol rate decreased by $1/L$ as (\ref{eq22})\cite{8907447}.

The total power consumption of the PT for the considered RIS-BD based SR system is composed of the power consumed by the power amplifier, which is modelled to be proportional to the signal transmission power ${\left\| {\mathbf{w}} \right\|^2}$, as well as the circuit power consumed in the PT, denoted as $P_{s}$. Therefore, the total power consumption of the PT is ${P_{PT}} \triangleq \mu {\left\| {\mathbf{w}} \right\|^2} + {P_s}$, where $\mu>1$ denotes the inefficiency of the power amplifier at the PT.
On the other hand, the RIS-BD does not consume power for signal transmission, since its reflectors are passive elements that do not actively transmit signal. The power dissipated at the RIS-BD with $N$ reflecting elements is modelled as ${P_{RIS-BD}} \triangleq N{P_r}$, where ${P_r}$ denotes the power consumption of each phase shifter.

We define the EE of the PT as the ratio of the primary communication rate to its power consumption, which can be expressed as
\begin{equation}\label{eq15}
{EE_{PT}} = \frac{R_s}{P_{PT}}
   = \frac{B{\mathbb{E}_c}\Big[ {{{\log }_2}\big( {1 + \frac{| {( {{{\mathbf{h}}^{\mathrm{H}}} + \sqrt \rho  {{\mathbf{f}}^{\mathrm{H}}}{\mathbf{\Phi G}}c} ){\mathbf{w}}} |^2}{{{\sigma ^2}}}} \big)} \Big]}{\mu {{{\| {\mathbf{w}} \|^2}} + P_s}}.
\end{equation}
On the other hand, the EE of the RIS-BD is defined as the ratio of the backscattering communication rate to its power consumption, which can be expressed as
\begin{equation}\label{eq25}
  EE_{RIS-BD} = \frac{B}{{LN{P_r}}}{\log _2}\Big( {1 + \frac{{L{{| {\sqrt \rho  {{\mathbf{f}}^{\mathrm{H}}}{\mathbf{\Phi Gw}}} |}^2}}}{{{\sigma ^2}}}} \Big).
\end{equation}

It is observed from (\ref{eq15})-(\ref{eq25}) that the EEs of both PT and RIS-BD critically depend on the transmit beamforming vector $\bf{w}$ and reflection matrix at the RIS-BD ${\mathbf{\Phi}}$. By varying $\bf{w}$ and ${\mathbf{\Phi}}$, the complete EE region can be obtained, which is defined as the union of all EE pairs as
\begin{equation}\label{eq27}
  {\cal H }= \bigcup\limits_{{{\left\| {\bf{w}} \right\|}^2} \le {P_{\max }},{\mathbf{\Phi}} } {\left\{ {\left( {EE_{PT},EE_{RIS-BD}} \right)} \right\}}.
\end{equation}


\section{Maximum Individual EE}\label{s3}

To explicitly reveal the fundamental relationship between active and passive communications, in this section, we first analyze the maximum individual EE for the PT and the RIS-BD to get some insights before characterizing the EE region.

\subsection{Maximum Individual EE of the PT}\label{s3-1}

If the objective is to maximize the EE of the PT without considering that of the RIS-BD, we have the following optimization problem
\underline{\textbf{P1}:}
\begin{equation}\label{eq40}
\begin{aligned}
  \mathop {\max }\limits_{{\mathbf{w}},{\mathbf{\Phi}}} \;  &
  {EE}_{PT} = \frac{B{\mathbb{E}_c}\left[ {{{\log }_2}\Big( {1 + \frac{| {( {{{\mathbf{h}}^{\mathrm{H}}} + \sqrt \rho  {{\mathbf{f}}^{\mathrm{H}}}{\mathbf{\Phi G}}c} ){\mathbf{w}}} |^2}{{{\sigma ^2}}}} \Big)} \right]}{{\mu {{{\| {\mathbf{w}} \|}^2}} + P_s}}\\
  \text{s.t.}\; & \text{C1}:\; {\left\| {\mathbf{w}} \right\|^2} \leq {P_{\max }}, \\
  \; & \text{C2}:\; \phi_n=e^{j \theta_n}, {\theta _n} \in [0, 2\pi) ,\forall n = 1,\cdots,N.  \\
\end{aligned}
\end{equation}
Note that to maximize the EE of the PT, we have to deal with the expectation of a logarithmic function with respect to the random RIS-BD symbol $c$, which makes it very hard to solve \underline{\textbf{P1}} directly. 

To address such issues, we substitute ${R_s}$ in (\ref{eq20}) with its upper bound and convert the problem into a more tractable problem.
By using Jensen's inequality, ${R_s}$ in (\ref{eq20}) is approximated by its upper bound to the concave logarithmic function\cite{9686018}, i.e.,
\begin{equation}\label{eq41}
  {R_s} \le {R_{s}^{\mathrm{ub}}}
   = B{\log _2}\Big( {1 + \frac{{{\mathbb{E}_c}\left[ {|({{\mathbf{h}}^{\mathrm{H}}} + \sqrt \rho  {{\mathbf{f}}^{\mathrm{H}}}{\mathbf{\Phi G}}c){\mathbf{w}}{|^2}} \right]}}{{{\sigma ^2}}}} \Big)  = B{\log _2}\Big( {1 + \frac{{{{\mathbf{w}}^{\mathrm{H}}}{\mathbb{E}_c}\left[ {{{\mathbf{h}}_{\mathrm{eq}}}( c ){\mathbf{h}}_{\mathrm{eq}}^{\mathrm{H}}( c )} \right]{\mathbf{w}}}}{{{\sigma ^2}}}} \Big).
\end{equation}
Define $\sqrt \rho {{\mathbf{f}}^{\mathrm{H}}} {\boldsymbol{\Phi}{\mathbf{G}}} = \sqrt \rho {{\boldsymbol{\phi}}^{\mathrm{H}}} {{\text{diag}} ({{\mathbf{f}}^{\mathrm{H}}}) {\mathbf{G}}} = {{\boldsymbol{\phi}}^{\mathrm{H}}} {\mathbf{M}}$ for convenience. Due to the fact that $c \sim\mathcal{CN}\left({0,1}\right)$, we can derive the average correlation matrix of ${\mathbf{h}}_{\mathrm{eq}}^{\mathrm{H}}(c)$ as
\begin{equation}\label{eq10}
  {\mathbb{E}_c}[ {{{\mathbf{h}}_{\mathrm{eq}}}(c){\mathbf{h}}_{\mathrm{eq}}^{\mathrm{H}}(c)} ]
  = {\mathbb{E}_c}\left[ {{({{\mathbf{h}}^{\mathrm{H}}} + {{\boldsymbol{\phi}}^{\mathrm{H}}} {\mathbf{M}}c)^{\mathrm{H}}}({{\mathbf{h}}^{\mathrm{H}}} + {{\boldsymbol{\phi}}^{\mathrm{H}}} {\mathbf{M}}c)} \right] = {\mathbf{h}}{{\mathbf{h}}^{\mathrm{H}}} + {{\mathbf{M}}^{\mathrm{H}}}{\boldsymbol{\phi}}{{\boldsymbol{\phi}} ^{\mathrm{H}}}{\mathbf{M}}.
  \end{equation}
Moreover, by defining $\widehat {\bf{h}} = \frac{{\bf{h}}}{\sigma}$ and $\widehat {\mathbf{M}} = \frac{{\mathbf{M}}}{\sigma } = \frac{{\sqrt \rho  {\text{diag}}({{\mathbf{f}}^{\mathrm{H}}}){\mathbf{G}}}}{\sigma }$ as the , ${R_{s}^{\mathrm{ub}}}$ is given by
\begin{equation}\label{eq42}
{R_{s}^{\mathrm{ub}}} = B{\log _2}\Big( {1 + \frac{{{\mathbf{w}}^{\mathrm{H}}}({\mathbf{h}}{{\mathbf{h}}^{\mathrm{H}}} + {{\mathbf{M}}^{\mathrm{H}}}{\boldsymbol{\phi}}{{\boldsymbol{\phi}} ^{\mathrm{H}}}{\mathbf{M}}){\mathbf{w}}}{{{\sigma ^2}}}} \Big) = B {\log _2}\Big( {1 + {|{\widehat{\mathbf{h}}^{\mathrm{H}}}{\mathbf{w}}|^2} + |{{\boldsymbol{\phi}} ^{\mathrm{H}}}\widehat {\mathbf{M}}{\mathbf{w}}|^2} \Big).
\end{equation}
Therefore, by replacing $R_s$ in (\ref{eq40}) with ${R_{s}^{\mathrm{ub}}}$, we can transform the problem as
\underline{\textbf{P1-1}:}
\begin{equation}\label{eq43}
\begin{aligned}
  \mathop {\max }\limits_{{\mathbf{w}},{\boldsymbol{\phi}}} \;  & EE_{PT}^{\mathrm{ub}}= \frac{{B{{\log }_2}\Big( {1 + {|{\widehat{\mathbf{h}}^{\mathrm{H}}}{\mathbf{w}}|^2} + |{{\boldsymbol{\phi}} ^{\mathrm{H}}}\widehat {\mathbf{M}}{\mathbf{w}}|^2} \Big)}}{\mu {{{\| {\mathbf{w}} \|}^2}} + P_s}\\
  \text{s.t.}\; & \text{C1}:\; {\left\| {\mathbf{w}} \right\|^2} \leq {P_{\max }}, \\
  \; & \text{C2}:\; 0\leq {\theta _n} < 2\pi ,\forall n = 1,\cdots,N.  \\
\end{aligned}
\end{equation}
In the following, AO algorithm is proposed to solve \underline{\textbf{P1-1}}, where phase shifts vector and the transmit beamforming vector are updated alternately with the other fixed.

\subsubsection{Transmit Beamforming Optimization of ${\mathbf{w}}$}

We first consider the transmit beamforming vector ${\mathbf{w}}$ optimization problem with given phase shifts vector ${\boldsymbol{\phi}}$. For convenience, we decompose the transmit beamforming vector as ${\bf{w}}=\sqrt p {\bf{v}}$, where $p={\left\|{\bf{w}} \right\|^2}$ is the transmit power and $\bf{v}$ denotes the transmit direction with ${\left\| {\bf{v}} \right\|} = 1$. Then, the subproblem for transmit beamforming vector optimization of \underline{\textbf{P1-1}} reduces to
\underline{\textbf{P1-2}:}
\begin{equation}\label{eq55}
\begin{aligned}
  \mathop {\max }\limits_{{\mathbf{v}},p} \;  & EE_{PT}^{\mathrm{ub}} = \frac{{B{{\log }_2}\Big(1 + p{{\mathbf{v}}^{\mathrm{H}}}(\widehat {\mathbf{h}}{{\widehat {\mathbf{h}}}^{\mathrm{H}}} + {{\widehat {\mathbf{M}}}^{\mathrm{H}}}{\boldsymbol{\phi}} {{\boldsymbol{\phi}}^{\mathrm{H}}}\widehat {\mathbf{M}}){\mathbf{v}}\Big)}}{{\mu p + {P_s}}}\\
  \text{s.t.}\; & \text{C1-1}:\; {\left\| {\mathbf{v}} \right\|} = 1, \\
  \; & \text{C1-2}:\; 0 \leq p \leq {P_{\max }}.  \\
\end{aligned}
\end{equation}
It is not difficult to see that with any given $p$, the optimal beamforming direction ${\bf{v}}^{\star}$ to \underline{\textbf{P1-2}} is obtained by solving the following optimization problem
\underline{\textbf{P1-3}:}
\begin{equation}\label{eq61}
  \mathop {\max}\limits_{{\bf{v}}, {\left\|{\bf{v}}\right\|}=1}\; R\left( {{\mathbf{F}},{\mathbf{v}}} \right)={{{\mathbf{v}}^{\mathrm{H}}}{\mathbf{Fv}}}
\end{equation}
where ${\mathbf{F}}=\widehat {\mathbf{h}}{{\widehat {\mathbf{h}}}^{\mathrm{H}}} + {{\widehat {\mathbf{M}}}^{\mathrm{H}}}{\boldsymbol{\phi}} {{\boldsymbol{\phi}}^{\mathrm{H}}}\widehat {\mathbf{M}}$. \underline{\textbf{P1-3}} is the standard Rayleigh quotient problem, whose optimal value is the largest eigenvalue with respect to the positive semidefinite matrix $\mathbf{F}$, i.e. $R{\left( {{\mathbf{F}},{\mathbf{v}}} \right)_{\max }} = {\lambda _{\max }}$ and the optimal solution ${\bf{v}}^{\star}$ is given by the normalized eigenvector of $\mathbf{F}$ corresponding to ${\lambda _{\max }}$.

By substituting $R{\left( {{\mathbf{F}},{\mathbf{v}}} \right)_{\max }}$ into \underline{\textbf{P1-2}}, the optimization problem reduces to finding the optimal transmit power $p^{\star}$ as
\underline{\textbf{P1-4}:}
\begin{equation}\label{eq45}
  \mathop {\max }\limits_{p,0 \le p \le {P_{\max }}} \; {EE}_{PT}^{\mathrm{ub}}\left(p\right) = \frac{{B{{\log }_2}\left( {1 +{\lambda _{\max }}p} \right)}}{{\mu p + {P_s}}}
\end{equation}
For this problem, the optimal solution is obtained by
\begin{equation}\label{eq46}
  {p^{\star}} = \left[ {\frac{{{\lambda _{\max }}{P_s} - \mu }}{{\mu {\lambda _{\max }}\mathrm{W}(\frac{{({\lambda _{\max }}{P_s} - \mu )}}{{\mu {\text{e}}}})}} - \frac{1}{{{\lambda _{\max }}}}} \right]_0^{{P_{\max }}},
\end{equation}
where $[a]_b^c = \min \left\{ {\max \{ {a,b} \},c} \right\}$\cite{9497709}. Therefore, the optimal solution of \underline{\textbf{P1-2}} can be denoted as ${{\mathbf{w}}^{\star}} = \sqrt {{p^{\star}}} {{\mathbf{v}}^{\star}}$.

\subsubsection{Phase Optimization of ${\boldsymbol{\theta}}$}

Next, with any given transmit beamforming vector ${\mathbf{w}}$, we consider the phase shifts vector ${\boldsymbol{\phi}}$ optimization problem. Note that the objective ${EE}_{PT}^{\mathrm{ub}}$ in \underline{\textbf{P1-1}} is a monotonically increasing function of $|{{\boldsymbol{\phi}} ^{\mathrm{H}}}\widehat {\mathbf{M}}{\mathbf{w}}|$, the subproblem for phase shifts vector optimization of \underline{\textbf{P1-1}} reduces to
\underline{\textbf{P1-5}:}
\begin{equation}\label{eq44}
\begin{aligned}
  \mathop {\max }\limits_{{\boldsymbol{\phi}}} \;  &
  |{{\boldsymbol{\phi}} ^{\mathrm{H}}}\widehat {\mathbf{M}}{\mathbf{w}}| \\
  \text{s.t.} \; & \text{C2}:\; \phi_n=e^{j \theta_n}, {\theta _n} \in [0, 2\pi) ,\forall n = 1,\cdots,N.  \\
\end{aligned}
\end{equation}
By defining ${\mathbf{x}} = {\widehat {\mathbf{M}}{\mathbf{w}}} $, the objective function in \underline{\textbf{P1-5}} achieves its maximum value 
\begin{equation}\label{eq28}
  {\gamma ^{\star}}( {\mathbf{w}} ) = \sum\limits_{n = 1}^N {|{x_n}|} = \sum\limits_{n = 1}^N {|\widehat {\mathbf{m}}_n^{\mathrm{H}}{\mathbf{w}}|}  = {{\| {\mathbf{x}} \|}_1} = {{\| {\widehat {\mathbf{M}}{\mathbf{w}}} \|}_1}
\end{equation}
with optimal solutions
\begin{equation}\label{eq29}
  {\theta _n^{\star}}  = \eta ^{\star} - \arg ({x_n})  = \eta ^{\star} - \arg (\widehat {\mathbf{m}}_n^{\mathrm{H}}{\mathbf{w}}) = \eta ^{\star} - \arg \Big(\frac{{\sqrt \rho  {f_n}^{\mathrm{H}}{\mathbf{g}}_n^{\mathrm{H}}{\mathbf{w}}}}{\sigma }\Big),
\end{equation}
where ${\eta ^{\star} }$ is a constant, ${x_n}$ and ${f_n}$ are the $n$-th elements of ${\mathbf{x}}$ and ${\mathbf{f}}$, $\widehat {\mathbf{m}}_n^{\mathrm{H}}$ and ${{\mathbf{g}}_n^{\mathrm{H}}}$ are the $n$-th row vector of $\widehat {\mathbf{M}}$ and ${\mathbf{G}}$. Then the optimal solution of \underline{\textbf{P1-5}} is ${\boldsymbol{\phi}}^{\star}={\left[ {{e^{j\theta _1^{\star}}},{e^{j\theta _2^{\star}}}, \cdots ,{e^{j\theta _N^{\star}}}} \right]^{\mathrm{H}}}$.

The algorithm is initialized with ${\bf{w}}^{(0)}$ and $\boldsymbol{\phi}^{(0)}$ chosen from feasible sets. In the $i$-th iteration, with given $\boldsymbol{\phi}^{(i-1)}$, we first solve problem \underline{\textbf{P1-2}} to design the optimal transmit beamforming vector ${\bf{w}}^{(i)}$. With ${\bf{w}}^{(i)}$, we solve problem \underline{\textbf{P1-5}} and then obtain the optimal $\boldsymbol{\phi}^{(i)}$ with respect to ${\bf{w}}^{(i)}$ according to (\ref{eq29}). In this way, ${\bf{w}}^{(i)}$ and $\boldsymbol{\phi}^{(i)}$ are optimized alternatively, which is summarized in Algorithm \ref{alg1}. 
\begin{algorithm}[H]
  \renewcommand{\algorithmicrequire}{\textbf{Input}:}
  \renewcommand{\algorithmicensure}{\textbf{Output}:}
  \renewcommand{\algorithmicif}{\quad \textbf{If}}
  \renewcommand{\algorithmicendif}{\quad \textbf{End if}}
  \renewcommand{\algorithmicelse}{\quad \textbf{Else}}
  \renewcommand{\algorithmicreturn}{\textbf{Return}}
  \caption{AO algorithm for solving \underline{\textbf{P1-1}}}
  \label{alg1}
  \begin{algorithmic}[1]
    \STATE \textbf{Initialization}: ${{\bf{w}}^{(0)}} = \sqrt {p^{(0)}} {\bf{v}}^{(0)}$, where $p^{(0)} = P_{\max}$, ${\bf{v}}^{(0)} = \frac{\widehat {\mathbf{h}}}{\|\widehat {\mathbf{h}}\|}$, random phase shifts ${{\boldsymbol{\phi}}^{(0)}}$, ${\mathbf{F}}^{(0)}$ and $i=0$;
    \STATE \textbf{Repeat}
      \STATE \quad Update $i=i+1$;
      \STATE \quad Compute the eigenvalue decomposition of ${\mathbf{F}}^{(i-1)}$, denote $\lambda_{\max}^{(i)}$ the largest eigenvalue of ${\mathbf{F}}^{(i-1)}$ and ${{\bf{v}}^{(i)}}$ the eigenvector of ${\mathbf{F}}^{(i-1)}$ corresponding to $\lambda_{\max}^{(i)}$;
      \STATE \quad Update ${p^{(i)}}$ by (\ref{eq46}) with $\lambda_{\max}^{(i)}$ and ${{\bf{w}}^{(i)}} = \sqrt{p^{(i)}}{\mathbf{v}}^{(i)}$;
      \STATE \quad Update ${\theta _n^{(i)}}$ by (\ref{eq29}) with ${{\bf{w}}^{(i)}}$ and ${\boldsymbol{\phi}}^{(i)}=e^{j {\boldsymbol{\theta}}^{(i)}}$;
      \STATE \quad Update ${\mathbf{F}}^{(i)} = \widehat {\mathbf{h}}{{\widehat {\mathbf{h}}}^{\mathrm{H}}} + {{\widehat {\mathbf{M}}}^{\mathrm{H}}}{\boldsymbol{\phi}}^{(i)} {{\boldsymbol{\phi}}^{{(i)}H}}\widehat {\mathbf{M}}$;
    \STATE \textbf{Until} The fractional increase of the objective value of \underline{\textbf{P1-1}} is below a certain threshold $\kappa$;
    \RETURN ${{\bf{w}}^{\star}}={{\bf{w}}^{\left(i\right)}}$ and ${{\boldsymbol{\phi}}^{\star}}={{\boldsymbol{\phi}}^{(i)}}$.
  \end{algorithmic}
\end{algorithm}
\begin{proposition}\label{Pro1}
  \underline{\textbf{P1-1}} converges when the AO algorithm is used as shown in Algorithm \ref{alg1}.
\end{proposition}
\begin{proof}
  Please refer to Appendix \ref{App3}.
\end{proof}
Denote the final solution to \underline{\textbf{P1}} ${{\mathbf{w}}_1^{\star}}$ and ${\boldsymbol{\phi}_1 ^{\star}}$.
Therefore, the maximum $EE_{PT}$, denoted by $\eta_{PT}^{(1)}$, and the resulting EE of the RIS-BD, i.e., $\eta _{RIS-BD}^{(1)}$, are obtained as
\begin{equation}\label{eq47}
  \eta _{PT}^{(1)} = \frac{{B{{\log }_2}\Big(1 + |{{\widehat {\mathbf{h}}}^{\mathrm{H}}}{{\mathbf{w}}_1^{\star}}|^2 + |{\boldsymbol{\phi}_1 ^{\star \mathrm{H}}}\widehat {\mathbf{M}}{{\mathbf{w}}_1^{\star}}|^2\Big)}}{{\mu {{\| {{{\mathbf{w}}_1^{\star}}} \|}^2} + {P_s}}}, \;
  \eta _{RIS-BD}^{(1)} = \frac{B}{{LN{P_r}}}{\log _2}\Big( 1 + L|{{\boldsymbol{\phi}}_1 ^{\star \mathrm{H}}\widehat {\mathbf{M}}{\mathbf{w}_1^{\star}}|^2} \Big).
\end{equation}

\subsection{Maximum Individual EE of the RIS-BD}\label{s3-2}
Next, we consider the problem to maximize the EE of the RIS-BD, without considering that of the PT. Based on (\ref{eq25}), the problem can be formulated as
\underline{\textbf{P2}:}
\begin{equation}\label{eq52}
\begin{aligned}
  \mathop {\max }\limits_{{\mathbf{w}},{\boldsymbol{\Phi }}}
  \; & EE_{RIS-BD} = \frac{B}{{LN{P_r}}}{\log _2}\Big( {1 + \frac{{\rho L{{| {  {{\mathbf{f}}^{\mathrm{H}}}{\mathbf{\Phi Gw}}} |}^2}}}{{{\sigma ^2}}}} \Big) \\
  \text{s.t.}\; & \text{C1}:\; {\left\| {\mathbf{w}} \right\|^2} \leq {P_{\max }}, \\
  \; & \text{C2}:\; \phi_n=e^{j \theta_n}, {\theta _n} \in [0, 2\pi) ,\forall n = 1,\cdots,N.
\end{aligned}
\end{equation}
Since logarithmic function is monotonically increasing, P2 is equivalent to
\underline{\textbf{P2-1}:}
\begin{equation}\label{eq53}
\begin{aligned}
  \mathop {\max }\limits_{{\mathbf{w}},{\boldsymbol{\phi}}} \;  & |{{\boldsymbol{\phi}} ^{\mathrm{H}}}\widehat {\mathbf{M}}{\mathbf{w}}|  \\
  \text{s.t.}\; & \text{C1}:\; {\left\| {\mathbf{w}} \right\|^2} \leq {P_{\max }}, \\
  \; & \text{C2}:\; \phi_n=e^{j \theta_n}, {\theta _n} \in [0, 2\pi) ,\forall n = 1,\cdots,N.
\end{aligned}
\end{equation}
Note that the variables ${\mathbf{w}}$ and ${\boldsymbol{\phi}}$ in \underline{\textbf{P2-1}} are coupled with each other, which is difficult to jointly optimized. One approach to solve it is that under the optimal solution of ${\boldsymbol{\phi}}^{\star}$ with respect to ${\mathbf{w}}$ in (\ref{eq29}), we can substitute the objective function with ${\gamma ^{\star}}( {\mathbf{w}} )$ in (\ref{eq28}) to transform \underline{\textbf{P2-1}} into an optimization problem only related to ${\mathbf{w}}$. Since it aims to maximize a convex function, we still need to use some techniques to transform it into a convex problem, in which way we can only obtain the suboptimal solution and the computational complexity is relatively high. Therefore, similar as \underline{\textbf{P1-1}} analyzed in Subsection \ref{s3-1}, we apply the AO algorithm to solve \underline{\textbf{P2-1}}.

Initialize ${\bf{w}}^{(0)}$ and $\boldsymbol{\phi}^{(0)}$ feasible for \underline{\textbf{P2-1}}. In the $i$-th iteration, with given $\boldsymbol{\phi}^{(i-1)}$, to design the optimal transmit beamforming vector ${\bf{w}}^{(i)}$, we first solve the following problem \underline{\textbf{P2-2}:} 
\begin{equation}\label{eq30}
  \mathop {\max }\limits_{{\bf{w}},{\|{\bf{w}}\|^2}\le P_{\max} }
  \; |{{\boldsymbol{\phi}} ^{\mathrm{H}}}\widehat {\mathbf{M}}{\mathbf{w}}|
\end{equation}
It can be verified that the maximum-ratio transmission (MRT) is the optimal transmit beamforming solution to \underline{\textbf{P2-2}}, i.e., ${{\mathbf{w}}^{(i)}} = \sqrt {{P_{\max }}} \frac{{{{\widehat {\mathbf{M}}}^{\mathrm{H}}}\boldsymbol{\phi}^{(i-1)} }}{{\left\| {{{\widehat {\mathbf{M}}}^{\mathrm{H}}}\boldsymbol{\phi}^{(i-1)} } \right\|}}  \triangleq  {{\mathbf{w}}_{MRT}^{(i)}}$. With ${\bf{w}}^{(i)}$, we then solve problem \underline{\textbf{P1-5}} and then obtain the optimal $\boldsymbol{\phi}^{(i)}$ with respect to ${\bf{w}}^{(i)}$ according to (\ref{eq29}).
In this way, $\boldsymbol{\phi}^{(i)}$ and ${\bf{w}}^{(i)}$ are optimized alternately until the convergence.

Denote the final solution as ${{\mathbf{w}}_2^{\star}}$ and ${\boldsymbol{\phi}_2 ^{\star}}$.
Therefore, the maximum $EE_{RIS-BD}$, denoted by $\eta _{RIS-BD}^{(2)}$, and the resulting EE of the PT, i.e., $\eta_{PT}^{(2)}$, are obtained as
\begin{equation}\label{eq2}
  \eta _{RIS-BD}^{(2)} = \frac{B}{{LN{P_r}}}{\log _2}\Big( 1 + L|{{\boldsymbol{\phi}}_2 ^{\star \mathrm{H}}\widehat {\mathbf{M}}{\mathbf{w}_2^{\star}}|^2} \Big), \;
  \eta _{PT}^{(2)} = \frac{{B{{\log }_2}\Big(1 + |{{\widehat {\mathbf{h}}}^{\mathrm{H}}}{{\mathbf{w}}_2^{\star}}|^2 + |{\boldsymbol{\phi}_2 ^{\star \mathrm{H}}}\widehat {\mathbf{M}}{{\mathbf{w}}_2^{\star}}|^2\Big)}}{{\mu {{\| {{{\mathbf{w}}_2^{\star}}} \|}^2} + {P_s}}}.
\end{equation}

Based on the above results, it is found that in each iteration, the scheme of optimizing phase vector $\boldsymbol{\phi}$ to maximize the EE of the PT and that of the RIS-BD is the same. However, solutions to these two individual EE maximization problems are different in terms of the power allocation $p$ and the active normalized beamforming vector $\mathbf{v}$.
Specifically, we take a single iteration for example.
To maximize the EE of the RIS-BD, the PT should use its maximum power $P_{\max}$ and direct the signal towards the RIS-BD via MRT over the equivalent cascade channel ${\boldsymbol{\phi}^\mathrm{H}}\widehat {\mathbf{M}}$ through RIS-BD.
By contrast, to maximize the EE of the PT, the optimal normalized transmit beamforming $\mathbf{v}$ is given by the dominant eigendirection of the combined channel consisting of the primary link ${\widehat {\mathbf{h}}}$ as well as the backscattering link ${\boldsymbol{\phi}^\mathrm{H}}\widehat {\mathbf{M}}$. Besides, the PT should not use the maximum power generally.
Such results demonstrate that there exists a nontrivial trade-off between maximizing ${EE}_{PT}$ and ${EE}_{RIS-BD}$. To investigate such a trade-off, we will characterize the EE region of the considered RIS-BD based SR system as defined in (\ref{eq27}). But before we do that, we will analyze the asymptotic performance of this SR system to get more insights.

\section{Asymptotic Performance Analysis}\label{s4}

In this section, we analyze the asymptotic performance of SR when the number of PT antennas $M$ or RIS-BD elements $N$ goes very large, to exploit the channel hardening effect for SR systems.

By using the Rician channel model, the PT-to-receiver link can be expressed as
\begin{equation}\label{eq31}
    \mathbf{h} = \sqrt{\beta_{T R}} \left(\sqrt{\frac{K_1}{K_1+1}} \mathbf{h}_{LoS} + \sqrt{\frac{1}{K_1+1}} \mathbf{h}_{NLoS}\right),
\end{equation}
which is composed of a deterministic LoS path and spatially uncorrelated NLoS path.
For a typical RIS-BD deployment, the PT-to-RIS-BD link $\mathbf{G} $ and the RIS-BD-to-receiver link $\mathbf{f}$ can be modelled by Rician fading composed of a deterministic line-of-sight (LoS) path and spatially correlated non-LoS (NLoS) path:
\begin{equation}\label{eq32}
  \begin{aligned}
    \mathbf{G} & = \sqrt{\beta_{T S}} \left(\sqrt{\frac{K_2}{K_2+1}} \mathbf{G}_{LoS}+\sqrt{\frac{1}{K_2+1}} \mathbf{R}_{T S}^{1 / 2} \mathbf{G}_{NLoS}\right), \\
    \mathbf{f} & = \sqrt{\beta_{S R}} \left(\sqrt{\frac{K_3}{K_3+1}} \mathbf{f}_{LoS}+\sqrt{\frac{1}{K_3+1}} \mathbf{R}_{S R}^{1 / 2} \mathbf{f}_{NLoS}\right),
  \end{aligned}
\end{equation}
where $\beta_{T S}$ and $\beta_{S R}$ are the large-scale path losses of the PT-to-RIS-BD link and the RIS-BD-to-receiver link, $K_2$ and $K_3$ are the Rician K-factors between the PT and RIS-BD and between RIS-BD and receiver, and $\mathbf{R}_{T S} \in$ $\mathbb{C}^{N \times N}$ and $\mathbf{R}_{S R} \in \mathbb{C}^{N \times N}$ are their spatial correlation matrices. Also, $\mathbf{G}_{LoS} \in \mathbb{C}^{N \times M}$ and $\mathbf{f}_{LoS} \in \mathbb{C}^{N \times 1}$ are the deterministic LoS components and $\mathbf{G}_{NLoS} \in \mathbb{C}^{N \times M}$ and $\mathbf{f}_{NLoS} \in \mathbb{C}^{N \times 1}$ are the NLoS components whose entries are i.i.d. complex Gaussian random variables with zero mean and unit variance.
Under the above models, we study the asymptotic behavior of the maximum EE of the PT and the RIS-BD, respectively.

\subsection{Asymptotic Analysis of Maximum EE of the PT}\label{s4-1}

As analyzed in Subsection \ref{s3-1}, we can formulate an optimization problem to maximize the upper bound of the individual EE of the PT in (\ref{eq41}) with respect to ${{\mathbf{w}},{\boldsymbol{\Phi }}}$. By decomposing the transmit beamforming vector as ${\bf{w}}=\sqrt p {\bf{v}}$, we can rewrite the upper bound EE of the PT as
\begin{equation}\label{eq36}
  EE_{PT}^{\mathrm{ub}}= \frac{{B{{\log }_2}\Big(1 + \frac{p}{{{\sigma ^2}}}{{\mathbf{v}}^{\mathrm{H}}}\mathbf{D}{\mathbf{v}}\Big)}}{{\mu p + {P_s}}},
\end{equation}
where $\mathbf{D}={\mathbf{h}}{{\mathbf{h}}^{\mathrm{H}}} + \rho {{\mathbf{G}}^{\mathrm{H}}}{{\mathbf{\Phi }}^{\mathrm{H}}}{\mathbf{f}}{{\mathbf{f}}^{\mathrm{H}}}{\mathbf{\Phi G}}$. Then, we have
\begin{equation}\label{eq37}
  \mathbf{D} = {\mathbf{h}}{{\mathbf{h}}^{\mathrm{H}}} + \rho \sum\nolimits_{n = 1}^N {\sum\nolimits_{n' = 1}^N {{f_n}f_{n'}^* {e^{j\left( {{\theta _{n'}} - {\theta _n}} \right)}}{\mathbf{g}}_n^* {\mathbf{g}}_{n'}^\mathrm{T}} },
\end{equation}
where ${{\mathbf{g}}_{n}^\mathrm{T}}$ is the $n$-th row vector of $\mathbf{G}$.

To analyze the asymptotic performance of maximal EE of the PT, we consider the extreme case when the PT has massive antennas, i.e., $M \gg 1 $. Moreover, we assume that all channels are i.i.d. Rayleigh fading channels as well as the RIS-BD has massive elements, i.e., $K_1 = K_2 = K_3 = 0$, ${\mathbf{R}}_{SR} = {\mathbf{R}}_{TS} = {{\mathbf{I}}_N}$, and $N \gg 1$ to obtain some insights.

\begin{lemma}\label{Lem1}
  For RIS-BD-based SR with massive reflecting elements under i.i.d. Rayleigh fading channels, the EE of the PT in (\ref{eq41}) approaches to
  \begin{equation}\label{eq59}
    EE_{PT}^{\mathrm{ub}} \to\frac{{B{{\log }_2}\Big(1 + \frac{p}{{{\sigma ^2}}}M \left( {{\beta _{TR}} + \rho N{\beta _{SR}}{\beta _{TS}}} \right)\Big)}}{{\mu p + {P_s}}}.
  \end{equation}
\end{lemma}
\begin{proof}
  Please refer to Appendix \ref{App4}.
\end{proof}

Then the maximal value of $EE_{PT}^{\mathrm{ub}}$ is obtained by
\underline{\textbf{P1-6}:} 
\begin{equation}\label{eq65}
   \mathop {\max }\limits_{0 \leq p \leq {P_{\max }}} \frac{{B{{\log }_2}\Big(1 + \frac{p}{{{\sigma ^2}}}M \left( {{\beta _{TR}} + \rho N{\beta _{SR}}{\beta _{TS}}} \right)\Big)}}{{\mu p + {P_s}}}
\end{equation}
Similar as \underline{\textbf{P1-4}}, we have the optimal solution to \underline{\textbf{P1-6}} as
\begin{equation}\label{eq67}
  {p^{\star}} = \left[ {\frac{{{\widetilde D }{P_s} - \mu }}{{\mu {\widetilde D }\mathrm{W}(\frac{{({\widetilde D }{P_s} - \mu )}}{{\mu {\text{e}}}})}} - \frac{1}{{{\widetilde D }}}} \right]_0^{{P_{\max }}},
\end{equation}
where $\widetilde D = {M({{\beta _{TR}} + \rho N{\beta _{SR}}{\beta _{TS}}})}/{{{\sigma ^2}}}$.

For the number of RIS-BD elements $N=256$, Fig. \ref{fig_4} plots $EE_{PT}^{\mathrm{ub}}$ versus the number of PT antennas $M$ for four different maximum transmit power $P_{\max}$. While (\ref{eq59}) was derived for asymptotic setup with $N \gg 1$, it is also applicable for the extreme case with no RIS-BD, i.e., $N = 0$, in which case the transmit beamforming is aimed at the primary link. It is observed from Fig. \ref{fig_4} that for the considered setup, when $P_{\max} \le 28 \text{dBm}$, the EE of the PT grows as $P_{\max} $ increases, while for $P_{\max} = 28 \text{dBm}$ and $30 \text{dBm}$, the same EE of the PT is achieved. This can be shown from the optimal power allocation (\ref{eq67}).
\vspace*{-.5\baselineskip}
\begin{figure}[H]
  \centering
  \includegraphics[width = .5\textwidth]{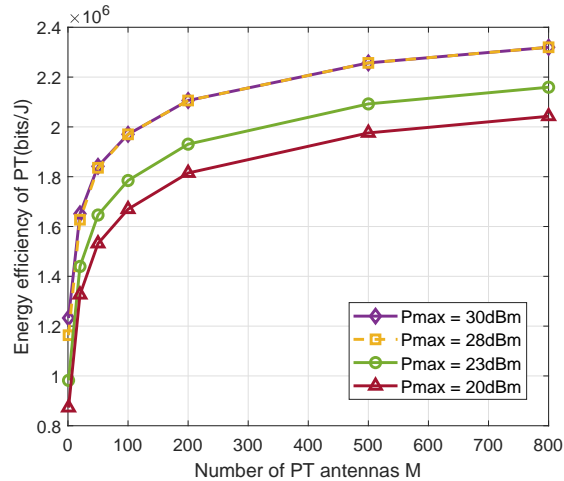}
  \vspace*{-.5\baselineskip}
  \caption{The EE of the PT versus the number of PT antennas $M$ with different $P_{\max}$ in MISO SR systems.}
  \vspace*{-.5\baselineskip}
  \label{fig_4}
\end{figure}

\subsection{Asymptotic Analysis of Maximum EE of the RIS-BD}\label{s4-2}

In order to obtain the tractable asymptotic performance analysis of EE of the RIS-BD, we consider the extreme case when the RIS-BD has massive elements, i.e., $N \gg 1 $. We first consider the special single-input single-output (SISO) SR setup for $M = 1$ to explore the effect of the number of RIS-BD elements on the EE of the RIS-BD, and then we analyze the general case of MISO setup.

Consider the optimization problem \underline{\textbf{P2}} whose goal is to maximize the individual EE of the RIS-BD with respect to ${{\mathbf{w}},{\boldsymbol{\Phi }}}$. For any given ${\boldsymbol{\Phi }}$, the optimal solution for $\mathbf{w}$ is the MRT beamforming: $\mathbf{w} = \sqrt{P_{\max}}\frac{\mathbf{G}^{{H}} \boldsymbol{\Phi}^{{H}} \mathbf{f}}{\|\mathbf{G}^{{H}} \boldsymbol{\Phi}^{{H}} \mathbf{f}\|}$, and the achievable rate of RIS-BD is $R_c = \frac{B}{L}{\log _2}\big( {1 + \frac{{\rho L {P_{\max}} }}{{{\sigma ^2}}}}{\|{{{\mathbf{f}}^{\mathrm{H}}}{\boldsymbol{\Phi}\mathbf{G}}} \|^2} \big)$, where
\begin{equation}\label{eq14}
  \|{{{\mathbf{f}}^{\mathrm{H}}}{\boldsymbol{\Phi}}{\mathbf{G}}} \|^2 =\sum\nolimits_{m=1}^M \left|\sum\nolimits_{n=1}^N {| f_n|| g_{n m}| e^{j\left(\theta_n+\arg( f_n^{*}+ g_{n m}) \right)}}\right|^2.
\end{equation}

\subsubsection{SISO SR}

In this case, the PT has only one antenna, i.e., $M = 1$. Therefore, $\|{{{\mathbf{f}}^{\mathrm{H}}}{\boldsymbol{\Phi}}{\mathbf{G}}} \|^2$ in (\ref{eq14}) reduces to
\begin{equation}\label{eq17}
  {\left| {{{\mathbf{f}}^{\mathrm{H}}}{\boldsymbol{\Phi}\mathbf{g}}} \right|^2} = {\left| {\sum\nolimits_{n = 1}^N {|{f_n}||{g_n}|{e^{j\left(\theta_n+\arg( f_n^{*}+ g_{n m}) \right)}}} } \right|^2},
\end{equation}
whose optimal RIS-BD phase shift is ${\theta _n^{\star}} = -\arg( f_n^{*}+ g_{n m}) $. Then, the EE of the RIS-BD can be written as
\begin{equation}\label{eq21}
  EE_{RIS-BD} = \frac{B}{{LN{P_r}}}{\log _2}\big( 1 + \frac{\rho L {P_{\max}}}{{\sigma ^2}}{\Big| {\sum\nolimits_{n = 1}^N {|{f_n}||{g_n}|} } \Big|^2}\big).
\end{equation}
For notational convenience, we define $X = \sum\nolimits_{n = 1}^N {|{f_n}||{g_n}|} $. Then, the random variable $X$ has the following property:
\begin{lemma}\label{Lem2}
  The mean of $X$ is $\mathbb{E} \left[ X \right] = \sum\nolimits_{n = 1}^N {{\mu _{f,n}}{\mu _{g,n}}}$, 
  where
  \begin{equation}\label{eq6}
    \begin{aligned}
      {\mu _{f,n}} =  \sqrt {\frac{{\beta _{SR}} \pi {\sum\limits_{r = 1}^N {\left| {{{({{\mathbf{R}}_{SR}})}_{n,r}}} \right|}} }{{4({K_3} + 1)}}}  \times {\mathrm{L}_{\frac{1}{2}}}\Big( { - \frac{{{K_3}}}{{\sum\limits_{r = 1}^N | {{({{\mathbf{R}}_{SR}})}_{n,r}}|}}} \Big), \\
      {\mu _{g,n}} =  \sqrt {\frac{{\beta _{TS}} \pi {\sum\limits_{r = 1}^N | {{({{\mathbf{R}}_{TS}})}_{n,r}}|}}{{4({K_2} + 1)}}}  \times {\mathrm{L}_{\frac{1}{2}}}\Big( { - \frac{{{K_2}}}{{\sum\limits_{r = 1}^N | {{({{\mathbf{R}}_{TS}})}_{n,r}}|}}} \Big). 
    \end{aligned}
  \end{equation}
\end{lemma}
\begin{proof}
  Please refer to Appendix \ref{App1}.
\end{proof}
\begin{corollary}\label{cor1}
  For SISO RIS-BD-based SR with massive reflecting elements, i.e., $N \gg 1$ and $\sum\nolimits_{r = 1}^N {|{{\left( {{{\mathbf{R}}_{SR}}} \right)}_{n,r}}|}  = \sum\nolimits_{r = 1}^N {|{{\left( {{{\mathbf{R}}_{TS}}} \right)}_{n,r}}|}  = 1$\cite{9357969}, the maximum EE of the RIS-BD in (\ref{eq21}) approaches to
  \begin{equation}\label{eq7}
    EE_{RIS-BD} \to \frac{B}{{LN{P_r}}}{\log _2}(1 + \frac{{\rho L{P_{\max }}}}{{{\sigma ^2}}}{N^2}\mu _f^2\mu _g^2),
  \end{equation}
  where
  \begin{equation}\label{eq23}
    {\mu _f} = \sqrt {\frac{{{\beta _{SR}}\pi }}{{4({K_3} + 1)}}}{\mathrm{L}_{\frac{1}{2}}}(-{K_3}), \;
    {\mu _g} = \sqrt {\frac{{\beta _{TS}} \pi }{{4({K_2} + 1)}}}{\mathrm{L}_{\frac{1}{2}}}(-{K_2}).
  \end{equation}
\end{corollary}
\begin{proof}
  Since $|f_n|$ and $|g_n|$ are statistically independent and follow Rayleigh distribution with mean values $\mu_f$ and $\mu_g$, respectively, we have $\mathbb{E} \left[ {\left| {{f_n}} \right|\left| {{g_n}} \right|} \right] = {\mu _f}{\mu _g}$. By using the fact that $\frac{1}{N}\sum\nolimits_{n = 1}^N {\left| {{f_n}} \right|\left| {{g_n}} \right|}  \to \mathbb{E} \left[ {\left| {{f_n}} \right|\left| {{g_n}} \right|} \right] = {\mu _f}{\mu _g}$ as $N \to \infty$, it follows that
  \begin{equation}\label{eq24}
    \begin{aligned}
      EE_{RIS-BD} & \to  \frac{B}{{LN{P_r}}}{\log _2}\Big(1 + \frac{{\rho L{P_{\max }}}}{{{\sigma ^2}}}{\left( {N\mathbb{E}\left[ {\left| {{f_n}} \right|\left| {{g_n}} \right|} \right]} \right)^2}\Big) = \frac{B}{{LN{P_r}}}{\log _2}(1 + \frac{{\rho L{P_{\max }}}}{{{\sigma ^2}}}{N^2}\mu _f^2\mu _g^2)\\
      & = \frac{B}{{LN{P_r}}}{\log _2}\Big(1 + \frac{{\rho L{P_{\max }}}}{{{\sigma ^2}}}\frac{{{N^2}{\pi ^2}{\beta _{SR}}{\beta _{TS}}\mathrm{L}_{\frac{1}{2}}^2(-{K_2})\mathrm{L}_{\frac{1}{2}}^2(-{K_3})}}{{16({K_2} + 1)({K_3} + 1)}}\Big).
    \end{aligned}
  \end{equation}
  This thus completes the proof.
\end{proof}
\begin{corollary}\label{cor2}
  For SISO RIS-BD-based SR with massive reflecting elements under i.i.d. Rayleigh fading channels, i.e., $N \gg 1$, $K_1 = K_2 = K_3 = 0$ and ${\mathbf{R}}_{SR} = {\mathbf{R}}_{TS} = {{\mathbf{I}}_N}$, the maximum EE of the RIS-BD in (\ref{eq21}) approaches to
  \begin{equation}\label{eq8}
    EE_{RIS-BD} \to \frac{B}{{LN{P_r}}}{\log _2}\Big(1 + \frac{{\rho L{P_{\max }}}}{{{\sigma ^2}}}\frac{{N^2}{{\pi ^2}{\beta _{SR}}{\beta _{TS}}}}{{16}}\Big).
  \end{equation}
\end{corollary}

Fig. \ref{fig_6} plots $EE_{RIS-BD}$ versus the number of RIS-BD element $N$ for four different maximum transmit power $P_{\max}$ in SISO SR systems. It is observed from Fig. \ref{fig_6} that when $N$ is sufficiently large, $EE_{RIS-BD}$ goes down as $N$ increases.However, when $N$ is relatively small, i.e., $N \le 50$, the trends of $EE_{RIS-BD}$ show great difference at different levels of $P_{\max}$. Specifically, when $P_{\max}$ is relatively large, i.e, $P_{\max} = 36$ or $40\text{dBm}$, $EE_{RIS-BD}$ decrease sharply with the increase of $N$. On the other hand, when $P_{\max} = 26$ or $30\text{dBm}$, $EE_{RIS-BD}$ will firstly increase and then decrease with the growth of $N$. This is not difficult to show by (\ref{eq8}) that the circuit power consumed by RIS-BD increases linearly with $N$, while the achievable rate of the RIS-BD shows logarithmic growth. Moreover, due to the fact that the RIS-BD does not provide power actively, it is apparent that the EE of the RIS-BD will increase as $P_{\max} $ goes large.
\vspace*{-1\baselineskip}
\begin{figure}[H]
  \centering
  \includegraphics[width = .48\textwidth]{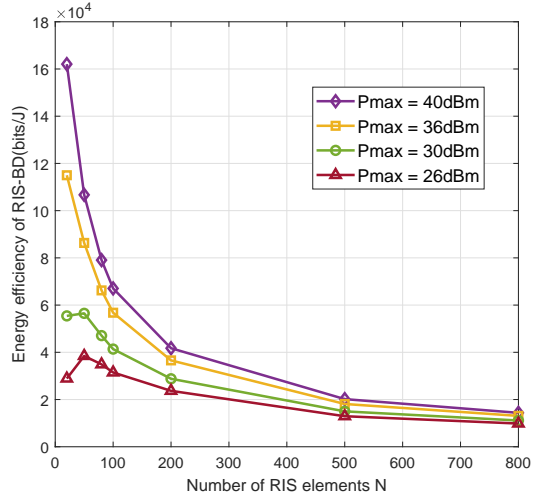}
  \vspace*{-.5\baselineskip}
  \caption{The EE of the RIS-BD versus the number of RIS-BD elements $N$ with different $P_{\max}$ in SISO SR systems.}
  \vspace*{-1\baselineskip}
  \label{fig_6}
\end{figure}

\subsubsection{MISO SR}

When the PT has multiple antennas and the RIS-BD has very large number of elements, i.e., $M > 1$ and $N \gg 1$, the EE of the RIS-BD can thus be written as
\begin{equation}\label{eq81}
  EE_{RIS-BD} = \frac{B}{{LN{P_r}}}{\log _2}\Big( 1 + \frac{\rho L {P_{\max}}}{{\sigma ^2}}{\|{{{\mathbf{f}}^{\mathrm{H}}}{\boldsymbol{\Phi}\mathbf{G}}} \|^2}\Big).
\end{equation}

\begin{lemma}\label{Lem3}
  Define $Y = \|{{{\mathbf{f}}^{\mathrm{H}}}{\boldsymbol{\Phi}\mathbf{G}}} \|^2$ in (\ref{eq14}), the random variable $Y$ thus follows non-central Chi-square distribution with $2M$ degrees of freedom, whose non-centrality parameter is
  \begin{equation}\label{eq26}
    \lambda  = \sum\nolimits_{m = 1}^M {{\Big( {\sum\nolimits_{n = 1}^N {\mu _{f,n}^*  {\mu _{g,n{m}}} {e^{j{\theta _n}}}} } \Big)^2}},
  \end{equation}
  where ${\mu _{f,n}} = \sqrt {\frac{{{\beta _{SR}}{K_3}}}{{{K_3} + 1}}} {f_{LoS,n}}$ and ${\mu _{g,n{m}}} = \sqrt {\frac{{{\beta _{TS}}{K_2}}}{{{K_2} + 1}}} {g_{LoS,n{m}}}$. 

  Therefore, the mean of $Y$ is $\mathbb{E}\left[ Y \right] =  \sum\nolimits_{m = 1}^M {\sum\nolimits_{n = 1}^N {\sigma _{f,n}^2\sigma _{g,n}^2 + \mu _{f,n}^2\sigma _{f,n}^2 + \mu _{g,nm}^2\sigma _{g,n}^2}}  + \lambda$, where $\sigma _{f,n}^2 = \frac{{{\beta _{SR}}}}{{{K_3} + 1}} \sum\nolimits_{r = 1}^N {|{{\left( {{{\mathbf{R}}_{SR}}} \right)}_{n,r}}|}$ and $\sigma _{g,n}^2 = \frac{{{\beta _{TS}}}}{{{K_2} + 1}}\sum\nolimits_{r = 1}^N {|{{\left( {{{\mathbf{R}}_{TS}}} \right)}_{n,r}}|}$.
\end{lemma}
\begin{proof}
  Please refer to Appendix \ref{App2}.
\end{proof}

To get more insight, we consider the special case of i.i.d. Rayleigh fading channels, for which we have the following Corollary \ref{cor3}.
\begin{corollary}\label{cor3}
  For MISO RIS-BD-based SR with massive reflecting elements under i.i.d. Rayleigh fading channels, i.e., $M \gg 1$, $N \gg 1$, $K_1 = K_2 = K_3 = 0$ and ${\mathbf{R}}_{SR} = {\mathbf{R}}_{TS} = {{\mathbf{I}}_N}$, the maximum EE of the RIS-BD in (\ref{eq81}) approaches to
  \begin{equation}\label{eq33}
    EE_{RIS-BD} \to \frac{B}{{LN{P_r}}}{\log _2}(1 + \frac{{\rho L{P_{\max }}}}{{{\sigma ^2}}}{MN{\beta _{SR}}{\beta _{TS}}}).
  \end{equation}
\end{corollary}

Fig. \ref{fig_7:a} and \ref{fig_7:b} plot the EE of the RIS-BD $EE_{RIS-BD}$ versus the number of RIS-BD element $N$ or the number of PT antennas $M$ in MISO SR systems. Comparing these two figures, we can clearly see that the increase of the number of PT antennas $M$ contributes to the improvement of $EE_{RIS-BD}$. On the contrary, the increase of the number of RIS-BD elements $N$ may compromise $EE_{RIS-BD}$.
\vspace*{-1\baselineskip}
\begin{figure}[H]
  \centering
  \subfigure[versus the number of RIS-BD elements $N$]{
  \label{fig_7:a}
  \includegraphics[width = .46\textwidth]{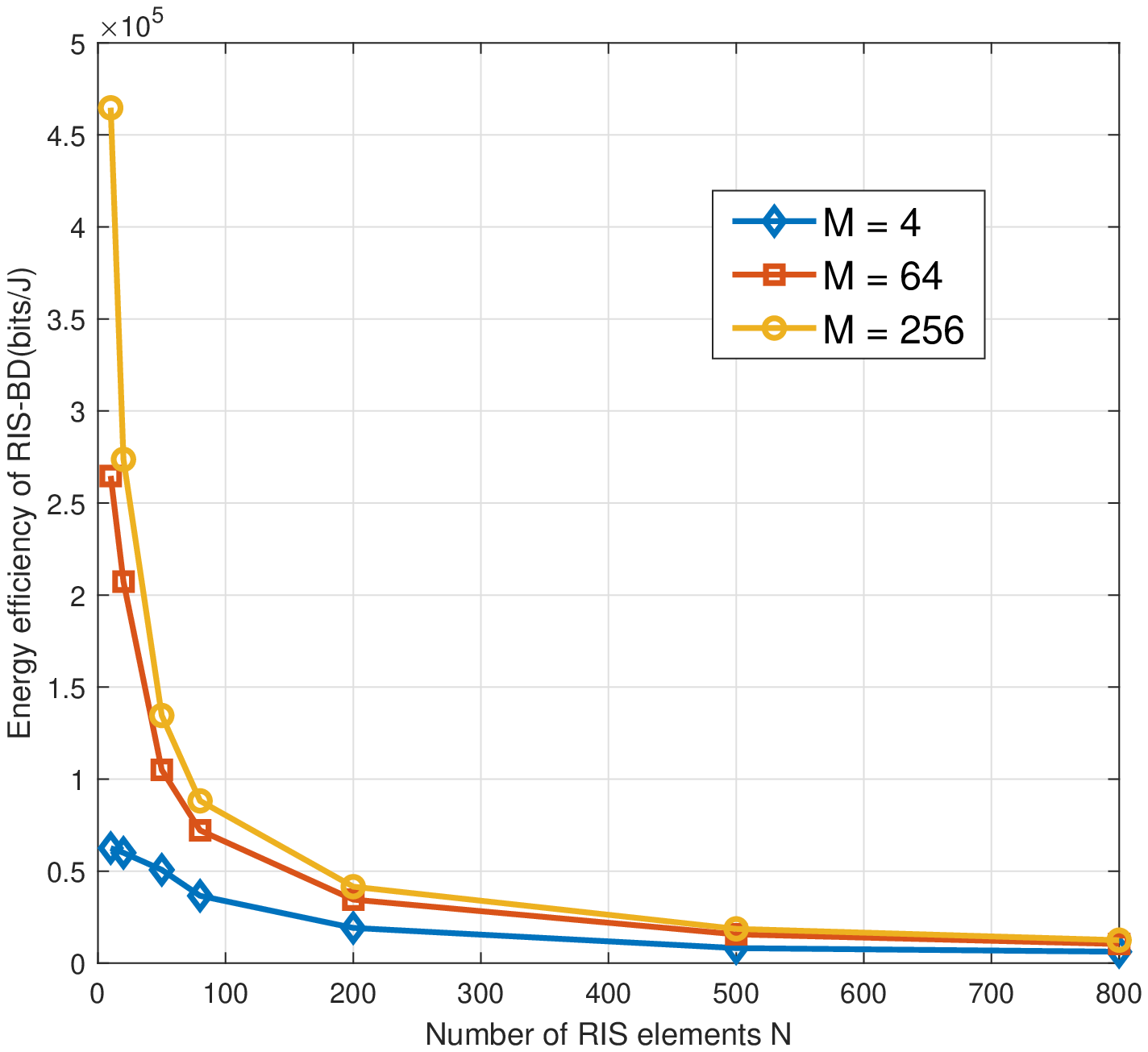}}
  \subfigure[versus the number of PT antennas $M$ ]{
  \label{fig_7:b}
  \includegraphics[width = .5\textwidth]{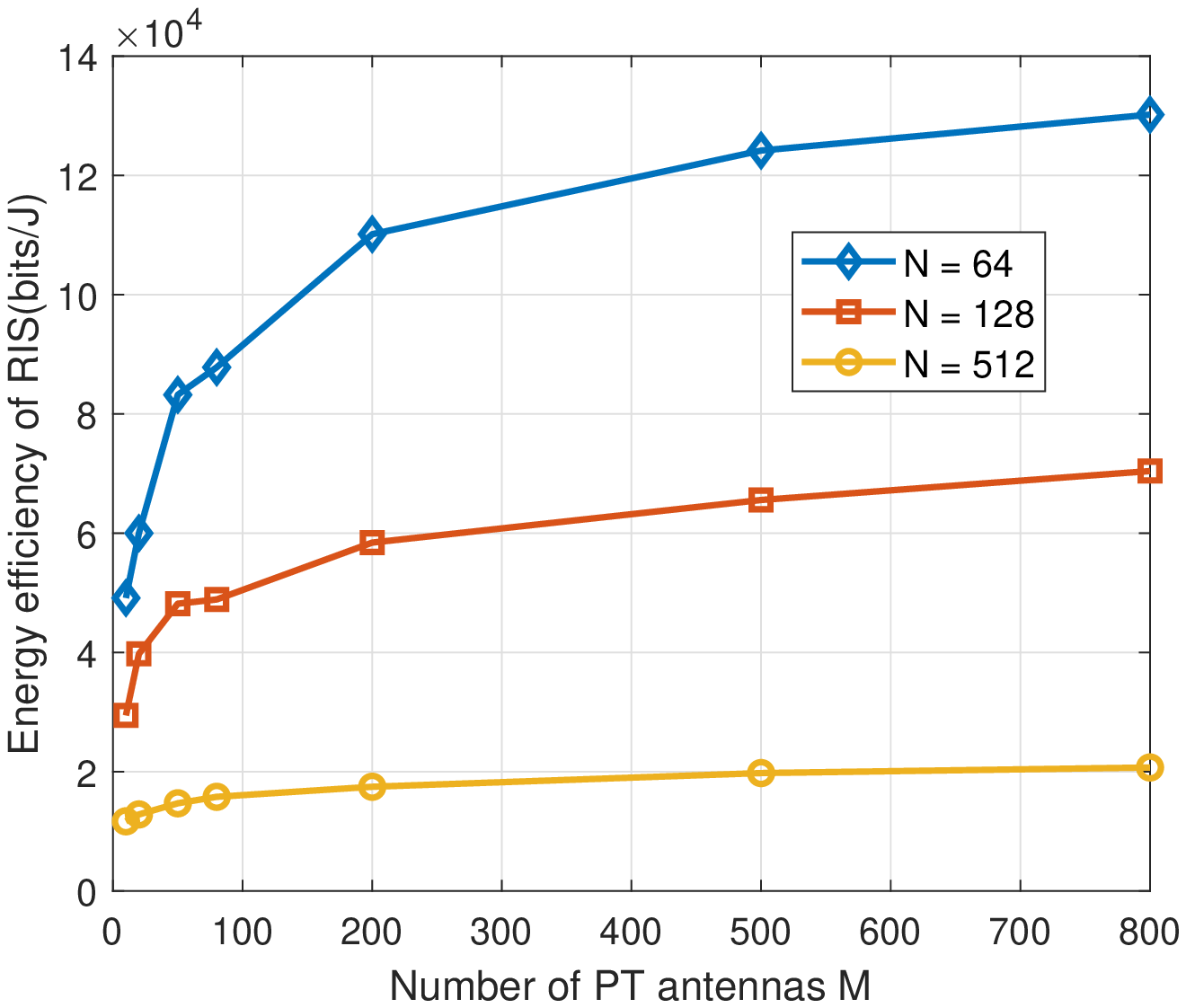}}
  \vspace*{-.5\baselineskip}
  \caption{The EE of the RIS-BD in MISO SR systems.}
  \vspace*{-1\baselineskip}
  \label{fig_7}
\end{figure}

\section{EE Region Characterization}\label{s5}

Of particular interest of the EE region in (\ref{eq27}) is its outer boundary, also called the Pareto boundary, which is defined as the union of all EE pairs $\left( {EE_{PT},EE_{RIS-BD}} \right)$ for which it is impossible to increase one without decreasing the other\cite{8316986, 9261117}.

\subsection{Problem Formulation and Transformation}\label{s5-1}
By following similar technique in \cite{10039158}, we can characterize the Pareto optimal EE pairs based on the concept of EE profile. Specifically, any EE pair on the Pareto boundary of the EE region ${{\cal H} }$ can be obtained via solving the following problem \underline{\textbf{P3}} with a given EE profile ${\boldsymbol{\alpha}} = \left( {\alpha ,1-\alpha }\right)$:
\begin{equation}\label{eq5}
\begin{aligned}
  \mathop {\max }\limits_{{\eta},{\mathbf{w}},{\boldsymbol{\Phi}}} & \quad \eta \\
  \text{s.t.}
  \; & \text{C1}:\;{\|{\bf{w}}\|^2}\le P_{\max}, \\
  \; & \text{C2}:\; \phi_n=e^{j \theta_n}, {\theta _n} \in [0, 2\pi) ,\forall n = 1,\cdots,N, \\
  \; & \text{C3}:\; \frac{B{\mathbb{E}_c}\left[ {{{\log }_2}\Big( {1 + \frac{| {( {{{\mathbf{h}}^{\mathrm{H}}} + \sqrt \rho  {{\mathbf{f}}^{\mathrm{H}}}{\mathbf{\Phi G}}c} ){\mathbf{w}}} |^2}{{{\sigma ^2}}}} \Big)} \right]}{{\mu {{{\| {\mathbf{w}} \|}^2}} + P_s}} \ge \alpha \eta, \\
  \; & \text{C4}:\; \frac{B}{{LN{P_r}}}{\log _2}\Big( {1 + \frac{{\rho L{{| {{{\mathbf{f}}^{\mathrm{H}}}{\mathbf{\Phi Gw}}} |}^2}}}{{{\sigma ^2}}}} \Big) \ge \left( {1 - \alpha } \right)\eta ,
\end{aligned}
\end{equation}
where $0 < \alpha < 1$ is the target ratio between $EE_{PT}$ and $ \eta $. By varying $\alpha$ between $0$ and $1$, the complete Pareto boundary of the EE region ${\cal H }$ can be characterized. For a given ${\boldsymbol{\alpha}}$, we denote ${\eta}^{\star}$ as the optimal value of \underline{\textbf{P3}}. Then ${{\eta}^{\star}}{\boldsymbol{\alpha}}$ is a Pareto optimal EE pair corresponding to the intersection between a ray in the direction of ${\boldsymbol{\alpha}}$ and the Pareto boundary of the EE region.
Notice that since the left hand side of $\text{C4}$ is monotonically increasing with respect to $\gamma_c = \frac{{\rho L{{| {{{\mathbf{f}}^{\mathrm{H}}}{\mathbf{\Phi Gw}}} |}^2}}}{{{\sigma ^2}}}$, we can recast it into a new constraint. For any fixed value $\eta$, with the definition before, \underline{\textbf{P3}} can be transformed to the following feasibility-check problem 
\underline{\textbf{P3-1}:}
\begin{equation}\label{eq35}
\begin{aligned}
  {\text{Find}} & \quad {\mathbf{w}},{\boldsymbol{\phi}} \\
  \text{s.t.}
  \; & \text{C1}:\;{\|{\bf{w}}\|^2}\le P_{\max}, \\
  \; & \text{C2}:\; \phi_n=e^{j \theta_n}, {\theta _n} \in [0, 2\pi) ,\forall n = 1,\cdots,N, \\
  \; & \text{C3-1}:\; \frac{B{\mathbb{E}_c}\left[ {{{\log }_2}\Big( {1 + | {({{\widehat {\mathbf{h}}^{\mathrm{H}}} + {{\boldsymbol{\phi}} ^{\mathrm{H}}}\widehat {\mathbf{M}} c} ){\mathbf{w}}} |^2} \Big)} \right]}{{\mu {{{\| {\mathbf{w}} \|}^2}} + P_s}} \ge \alpha \eta, \\
  \; & \text{C4-1}:\; {|{{{\boldsymbol{\phi}}^{\mathrm{H}}} \widehat{\mathbf{M}} {\mathbf{w}}} |^2} \ge {\gamma _{th}} ,
\end{aligned}
\end{equation}
where ${\gamma _{th}} = \frac{1}{L}\big( {{2^{\frac{{\left( {1 - \alpha } \right)\eta LN{P_r}}}{B}}} - 1} \big)$.
If $\eta $ is feasible to \underline{\textbf{P3-1}}, then the optimal value of \underline{\textbf{P3}} satisfies ${\eta}^{\star} \ge \eta $; otherwise, ${\eta}^{\star} < \eta $. Thus, by solving \underline{\textbf{P3-2}} with different $\eta $ and applying the efficient bisection method, \underline{\textbf{P3}} can be solved.
It is noted that \underline{\textbf{P3-1}} is feasible if and only if $ {|{{{\boldsymbol{\phi}}^{\star \mathrm{H}}} \widehat{\mathbf{M}} {\mathbf{w}^{\star}}} |^2} \ge {\gamma _{th}}$, where ${\boldsymbol{\phi}}^{\star}$ and ${{\mathbf{w}}^{\star}}$ are the optimal solutions of the following optimization problem
\underline{\textbf{P3-2}:}
\begin{equation}\label{eq66}
\begin{split}
  \mathop {\max }\limits_{{\bf{w}},{\boldsymbol{\phi}}}
  \;& {|{{{\boldsymbol{\phi}}^{\mathrm{H}}} \widehat{\mathbf{M}} {\mathbf{w}}} |^2} \\
  {\text{s.t.}}
  \; & \text{C1}, \text{C2}, \text{C3-1}.
\end{split}
\end{equation}

It is challenging to directly obtain the optimal solution of \underline{\textbf{P3-2}} due to the non-convex constraints with respect to ${\mathbf{w}}$ and ${\boldsymbol{\theta}}$, and even worse, they are coupled together. Moreover, another difficulty lies in that the left hand side of the constraint $\text{C3-1}$ involves an expectation with respect to the random RIS-BD symbol $c$. To tackle this issue, we approximate it by its sample average, as elaborated in the following sample-average based approach.

\subsection{Sample-Average Based Approach}\label{s5-2}

Notice that the main difference between \underline{\textbf{P3-2}} in (\ref{eq66}) and \underline{\textbf{P2-1}} in (\ref{eq53}) is the additional constraint $\text{C3-1}$. It is not difficult to find that when $\alpha \eta$ is small enough, i.e., $\alpha \eta \le {\overline {EE} _{PT}}$, where
\begin{equation}\label{eq49}
  {\overline {EE} _{PT}} = \frac{{B{\mathbb{E}_c}\left[ {{{\log }_2}\Big(1 + {|({\widehat{\mathbf{h}}^{\mathrm{H}}} + {\boldsymbol{\phi}}_2^{\star \mathrm{H}} \widehat {\mathbf{M}}c){\mathbf{w}}_2^{\star}{|^2}}\Big)} \right]}}{{\mu {{\| {{\mathbf{w}}_2^{\star}} \|}^2} + {P_s}}},
\end{equation}
the solution to \underline{\textbf{P3-2}} can be obtained as ${\boldsymbol{\phi}_2 ^{\star}}$ and ${{\mathbf{w}}_2^{\star}}$ stated in Subsection \ref{s3-2}.
With the above discussion, the remaining task for solving \underline{\textbf{P3-2}} is to consider the case $\alpha \eta > {\overline {EE} _{PT}}$.

For the sample-average based approach, the expectation of the primary transmission rate in (\ref{eq20}) is approximated by its sample average. Specifically, we assume that $c_t, t=1,\cdots,T$ are $T$ independent realizations of $c$ following its distribution $c \sim\mathcal{CN}\left({0,1}\right)$. Then when $T$ is sufficiently large, based on the law of large numbers, ${\overline {EE} _{PT}}$ can be approximated as
\begin{equation}\label{eq50}
  {\overline {EE} _{PT}^{\text{samp}}} = \frac{1}{T}\sum\nolimits_{t = 1}^T {\frac{{B{{\log }_2}\Big(1 + |({{\widehat {\mathbf{h}}}^{\mathrm{H}}} + \phi {{_2^{\star}}^{\mathrm{H}}}\widehat {\mathbf{M}}{c_t}){\mathbf{w}}_2^{\star}{|^2}\Big)}}{{\mu {{\left\| {{\mathbf{w}}_2^{\star}} \right\|}^2} + {P_s}}}} .
\end{equation}
And then, for the case $\alpha \eta > {\overline {EE} _{PT}^{\text{samp}}}$, we need to consider the optimization problem
\underline{\textbf{P4}:}
\begin{equation}\label{eq54}
\begin{aligned}
  \mathop {\max }\limits_{{\bf{w}},{\boldsymbol{\phi}}}
  \;& {|{{{\boldsymbol{\phi}}^{\mathrm{H}}} \widehat{\mathbf{M}} {\mathbf{w}}} |^2} \\
  \text{s.t.}
  \; & \text{C1}:\;{\|{\bf{w}}\|^2}\le P_{\max}, \\
  \; & \text{C2}:\; \phi_n=e^{j \theta_n}, {\theta _n} \in [0, 2\pi) ,\forall n = 1,\cdots,N, \\
  \; & \text{C5}:\; \frac{1}{T}\sum\nolimits_{t = 1}^T {\frac{{B{{\log }_2}\Big(1 + |({{\widehat {\mathbf{h}}}^{\mathrm{H}}} + {\boldsymbol{\phi}}^{\mathrm{H}} \widehat {\mathbf{M}}{c_t}){\mathbf{w}}{|^2}\Big)}}{{\mu {{\left\| {{\mathbf{w}}} \right\|}^2} + {P_s}}}} \ge \alpha \eta. \\
\end{aligned}
\end{equation}
We may apply the AO algorithm to decouple \underline{\textbf{P4}} into several subproblems as follows.

\subsubsection{Transmit Beamforming Optimization of ${\mathbf{w}}$}

First, consider the transmit beamforming vector optimization problem with given phase shifts vector. For convenience, we define ${\mathbf{h}}_0^{\mathrm{H}} = {{\boldsymbol{\phi}} ^{\mathrm{H}}}\widehat {\mathbf{M}}$ and ${\mathbf{h}}_t^{\mathrm{H}} = {\widehat {\mathbf{h}}^{\mathrm{H}}} + {{\boldsymbol{\phi}} ^{\mathrm{H}}}\widehat {\mathbf{M}}{c_t}, t=1,\cdots,T$. In the $i$-th iteration, we obtain the optimal solution of ${\mathbf{w}}^{(i)}$ with given ${\boldsymbol{\phi}}^{(i-1)}$ by solving problem
\underline{\textbf{P4-1}:}
\begin{equation}\label{eq57}
\begin{aligned}
  \mathop {\max }\limits_{{\bf{w}}}
  \;& {|{{\mathbf{h}}_0^{\mathrm{H}} {\mathbf{w}}} |^2} \\
  \text{s.t.}
  \; & \text{C1}:\;{\|{\bf{w}}\|^2}\le P_{\max}, \\
  \; & \text{C5-1}:\; \frac{B}{T}\sum\nolimits_{t = 1}^T {{{\log }_2}\Big(1 + |{\mathbf{h}}_t^{\mathrm{H}}{\mathbf{w}}{|^2}\Big)}  \ge \alpha \eta \Big( {\mu {{\| {\mathbf{w}} \|}^2} + {P_s}} \Big). \\
\end{aligned}
\end{equation}
\underline{\textbf{P4-1}} is non-convex as the objective function is a convex function with respect to ${\bf{w}}$, the maximization of which is a non-convex optimization problem. Moreover, the constraint $\text{C5-1}$ is also non-convex, which is difficult to address.

To tackle this issue, we apply the SCA technique\cite{SCA} to transforms the non-convex optimization problem into a series of convex optimization problems, with guaranteed convergence to a Karush-Kuhn-Tucker (KKT) solution under some mild conditions. Specifically, consider the current $(k+1)$-th SCA iteration, in which the local point ${\bf{w}}^{(k)}$ is obtained in the previous iteration. By using the fact that any convex differentiable function is globally lower-bounded by its first-order Taylor expansion, we have the lower bound at this given local point ${\bf{w}}^{(k)}$:
\begin{equation}\label{eq72}
  {H_t}({\mathbf{w}}) =  |{\mathbf{h}}_t^{\mathrm{H}}{\mathbf{w}}{|^2}
  = {{\mathbf{w}}^{\mathrm{H}}}{{\mathbf{H}}_t}{\mathbf{w}}  
  \geq  {{\mathbf{w}}^{(k)}}^{\mathrm{H}}{{\mathbf{H}}_t}{{\mathbf{w}}^{(k)}} + 2\operatorname{Re} \Big[ {{{\mathbf{w}}^{(k)}}^{\mathrm{H}}{{\mathbf{H}}_t}({\mathbf{w}} - {{\mathbf{w}}^{(k)}})} \Big] 
  \triangleq  H_t^{{\text{lb}}}({\mathbf{w}}|{{\mathbf{w}}^{(k)}}),
\end{equation}
where ${{\bf{H}}_t} = {{\mathbf{h}}_t}{ {\mathbf{h}}_t^{\mathrm{H}}} $ and $t = 0,1,\cdots,T$. Therefore, by replacing with these global lower bounds in (\ref{eq72}), we have the following optimization problem
\underline{\textbf{P4-2}:}
\begin{equation}\label{eq62}
\begin{aligned}
  \mathop {\max }\limits_{{\bf{w}}}
  \; & H_0^{{\text{lb}}}({\mathbf{w}}|{{\mathbf{w}}^{(k)}}) \\
  \text{s.t.}
  \; & \text{C1-1}:\;{\|{\bf{w}}\|^2} - P_{\max} \le 0, \\
  \; & \text{C5-2}:\; \frac{B}{T}\sum\nolimits_{t = 1}^T {{{\log }_2}\Big(1 + H_t^{{\text{lb}}}({\mathbf{w}}|{{\mathbf{w}}^{(k)}})\Big)} \ge \alpha \eta \Big( {\mu {{\| {\mathbf{w}} \|}^2} + {P_s}} \Big). \\
\end{aligned}
\end{equation}
$\text{C1-1}$ and $\text{C5-2}$ are all convex sets, and the objective function of \underline{\textbf{P4-2}} is an affine function for any given local point ${\bf{w}}^{(k)}$. Thus \underline{\textbf{P4-2}} is a convex problem, which can be efficiently solved by standard convex optimization techniques or existing software tools such as CVX \cite{cvx}. Thanks to the global lower/upper bounds relationships in (\ref{eq72}), the optimal value of \underline{\textbf{P4-2}} gives at least a lower bound to that of \underline{\textbf{P4-1}}. By successively updating the local point ${\bf{w}}^{(k)}$ and solving \underline{\textbf{P4-2}}, a monotonically non-decreasing objective value of \underline{\textbf{P4-1}} can be obtained.

\subsubsection{Phase Optimization of ${\boldsymbol{\phi}}$}

Next, we focus on optimizing the phase shifts vector with fixed transmit beamforming vector.
In the $i$-th iteration, we obtain the optimal solution of ${\boldsymbol{\phi}}^{(i)}$ with given ${\mathbf{w}}^{(i)}$ by solving the following problem
\underline{\textbf{P4-3}:}
\begin{equation}\label{eq63}
\begin{aligned}
  \mathop {\max }\limits_{{\boldsymbol{\phi}}}
  \; & {|{{{\boldsymbol{\phi}}^{\mathrm{H}}} \widehat{\mathbf{M}} {\mathbf{w}}} |^2} \\
  \text{s.t.}
  \; & \text{C2}:\; \phi_n=e^{j \theta_n}, {\theta _n} \in [0, 2\pi) ,\forall n = 1,\cdots,N, \\
  \; & \text{C5}:\; \frac{1}{T}\sum\nolimits_{t = 1}^T {\frac{{B{{\log }_2}\Big(1 + |({{\widehat {\mathbf{h}}}^{\mathrm{H}}} + {\boldsymbol{\phi}}^{\mathrm{H}} \widehat {\mathbf{M}}{c_t}){\mathbf{w}}{|^2}\Big)}}{{\mu {{\left\| {{\mathbf{w}}} \right\|}^2} + {P_s}}}} \ge \alpha \eta. \\
\end{aligned}
\end{equation}
Before optimization ${\boldsymbol{\phi}}$, we first define $h = {\widehat {\mathbf{h}}^{\mathrm{H}}}{\mathbf{w}}$, ${{\boldsymbol{\beta }}_t} = \widehat {\mathbf{M}}{c_t}{\mathbf{w}}, t = 1, \cdots ,T$ for convenience. After that, we can rewrite $|({{\widehat {\mathbf{h}}}^{\mathrm{H}}} + {\boldsymbol{\phi}}^{\mathrm{H}} \widehat {\mathbf{M}}{c_t}){\mathbf{w}}{|^2}$ as
\begin{equation}\label{eq64}
\begin{aligned}
  & |({\widehat {\mathbf{h}}^{\mathrm{H}}} + {{\boldsymbol{\phi}} ^{\mathrm{H}}}\widehat {\mathbf{M}}{c_t}){\mathbf{w}}|^2
  = |h + {{\boldsymbol{\phi}} ^{\mathrm{H}}}{{\boldsymbol{\beta }}_t}|^2 = |{h^{\star}} + {\boldsymbol{\beta }}_t^{\mathrm{H}} {\boldsymbol{\phi}} |^2 
  = {( {h + {{\boldsymbol{\phi}} ^{\mathrm{H}}}{{\boldsymbol{\beta }}_t}} )}( {{h^{\star}} + {\boldsymbol{\beta }}_t^{\mathrm{H}} {\boldsymbol{\phi}} } ) \\
  = & h{h^{\star}} + h{\boldsymbol{\beta }}_t^{\mathrm{H}} {\boldsymbol{\phi}}  + {{\boldsymbol{\phi}} ^{\mathrm{H}}}{{\boldsymbol{\beta }}_t}{h^{\star}} + {{\boldsymbol{\phi}} ^{\mathrm{H}}}{{\boldsymbol{\beta }}_t}{\boldsymbol{\beta }}_t^{\mathrm{H}} {\boldsymbol{\phi}}  
  = {| h |^2} + 2\operatorname{Re} [ {{h}{\boldsymbol{\beta }}_t^{\mathrm{H}} {\boldsymbol{\phi}} } ] + {{\boldsymbol{\phi}} ^{\mathrm{H}}}{{\mathbf{B}}_t} {\boldsymbol{\phi}},
\end{aligned}
\end{equation}
where ${{\mathbf{B}}_t} = {\boldsymbol{\beta }}_t {\boldsymbol{\beta }}_t^{\mathrm{H}}$.
Recalling that ${\mathbf{x}} = {\widehat {\mathbf{M}}{\mathbf{w}}} $ defined before, \underline{\textbf{P4-3}} can be reformulated as
\underline{\textbf{P4-4}:}
\begin{equation}\label{eq70}
\begin{aligned}
  \mathop {\max }\limits_{{\boldsymbol{\phi}}}
  \; & |{{\boldsymbol{\phi}}^{\mathrm{H}}}{\mathbf{x}}|^2 \\
  \text{s.t.}
  \; & \text{C2-1}:\; | {{\phi _n}} | = 1,n = 1, \cdots ,N, \\
  \; & \text{C5-3}:\;  \alpha \eta ( {\mu {{\| {\mathbf{w}} \|}^2} + {P_s}} ) - \frac{B}{T}\sum\nolimits_{t = 1}^T {{{\log }_2}\Big(1 + {| h |^2} + 2\operatorname{Re} [ {{h}{\boldsymbol{\beta }}_t^{\mathrm{H}} {\boldsymbol{\phi}} } ] + {{\boldsymbol{\phi}} ^{\mathrm{H}}}{{\mathbf{B}}_t} {\boldsymbol{\phi}}\Big)}\le 0, \\
\end{aligned}
\end{equation}
where ${\phi_n}$ is the $n$-th elements of ${\boldsymbol{\phi}}$. To handle the non-convex constraint $\text{C2-1}$, we can loosen this constraint and rewrite \underline{\textbf{P4-4}} as
\underline{\textbf{P4-5}:}
\begin{equation}\label{eq73}
\begin{aligned}
  \mathop {\max }\limits_{{\boldsymbol{\phi}}}
  \; & |{{\boldsymbol{\phi}}^{\mathrm{H}}}{\mathbf{x}}|^2   \\
  \text{s.t.}
  \; & \text{C2-2}:\;  | {{\phi_n}} | \le 1,n = 1, \cdots ,N, \\
  \; & \text{C5-3}:\; \alpha \eta ( {\mu {{\| {\mathbf{w}} \|}^2} + {P_s}} ) - \frac{B}{T}\sum\nolimits_{t = 1}^T {{{\log }_2}\Big(1 + {| h |^2} + 2\operatorname{Re} [ {{h}{\boldsymbol{\beta }}_t^{\mathrm{H}} {\boldsymbol{\phi}} } ] + {{\boldsymbol{\phi}} ^{\mathrm{H}}}{{\mathbf{B}}_t} {\boldsymbol{\phi}}\Big)} \le 0. \\
\end{aligned}
\end{equation}
In order to solve this non-convex problem \underline{\textbf{P4-5}}, we also utilize the SCA technique. Defining $|{{\boldsymbol{\phi}} ^{\mathrm{H}}}{\mathbf{x}}|^2 = {{\boldsymbol{\phi}} ^{\mathrm{H}}}{\mathbf{x}}{{\mathbf{x}}^{\mathrm{H}}}{\boldsymbol{\phi}} = {{\boldsymbol{\phi}} ^{\mathrm{H}}}{{\mathbf{B}}_0} {\boldsymbol{\phi}} $, we have the global lower bound at a given local point ${\boldsymbol{\phi}}^{(k)}$ :
\begin{equation}\label{eq74}
  {B_t}\left( {\boldsymbol{\phi}}  \right)
  =  {{\boldsymbol{\phi}} ^{\mathrm{H}}}{{\mathbf{B}}_t}{\boldsymbol{\phi}} 
  \geq  {{\boldsymbol{\phi}}^{(k)}}^{\mathrm{H}}{{\mathbf{B}}_t}{{\boldsymbol{\phi}}^{(k)}} + 2\operatorname{Re} \Big[ {{{\boldsymbol{\phi}}^{(k)}}^{\mathrm{H}}{{\mathbf{B}}_t}({\boldsymbol{\phi}} - {{\boldsymbol{\phi}}^{(k)}})} \Big] 
  \triangleq  B_t^{{\text{lb}}}({\boldsymbol{\phi}}|{{\boldsymbol{\phi}}^{(k)}}),
\end{equation}
for all $t=0,1,\cdots,T$, where the local point ${\boldsymbol{\phi}}^{(k)}$ is obtained in the previous $k$-th iteration.
With the above approximations, the non-convex problem \underline{\textbf{P4-5}} can be formulated in the following convex problem
\underline{\textbf{P4-6}:}
\begin{equation}\label{eq75}
\begin{aligned}
  \mathop {\max }\limits_{{\boldsymbol{\phi}}}
  \; & B_0^{{\text{lb}}}({\boldsymbol{\phi}}|{{\boldsymbol{\phi}}^{(k)}}) \\
  \text{s.t.}
  \; & \text{C2-2}: \; | {{\phi_n}} | \le 1,n = 1, \cdots ,N, \\
  \; & \text{C5-4}: \; \alpha \eta \Big( {\mu {{\| {\mathbf{w}} \|}^2} + {P_s}} \Big) - \frac{B}{T}\sum\nolimits_{t = 1}^T {{{\log }_2}\Big(\widetilde h  + 2\operatorname{Re} [ {{h}{\boldsymbol{\beta }}_t^{\mathrm{H}} {\boldsymbol{\phi}} } ] + B_t^{{\text{lb}}}({\boldsymbol{\phi}}|{{\boldsymbol{\phi}}^{(k)}})\Big)} \le 0, \\
\end{aligned}
\end{equation}
where $\widetilde h = 1 + |h{|^2}$.
Therefore, \underline{\textbf{P4-3}} can be solved by applying the SCA technique, where the approximated convex problem \underline{\textbf{P4-6}} is solved at each iteration by CVX \cite{cvx} or other techniques. 
\begin{algorithm}[H]
  \renewcommand{\algorithmicrequire}{\textbf{Input}:}
  \renewcommand{\algorithmicensure}{\textbf{Output}:}
  \renewcommand{\algorithmicif}{\quad \textbf{If}}
  \renewcommand{\algorithmicendif}{\quad \textbf{End if}}
  \renewcommand{\algorithmicelse}{\quad \textbf{Else}}
  \renewcommand{\algorithmicreturn}{\textbf{Return}}
  \caption{Bisection Method for solving \underline{\textbf{P3}}}
  \label{alg2}
  \begin{algorithmic}[1]
    \STATE \textbf{Initialization}: ${\eta_{L}}=0$, $\eta_ {U}$ to a sufficiently large number;
    \STATE \textbf{Repeat}
      \STATE \quad Let $\eta = {\left({{\eta_L}+{\eta_U}} \right)/2}$;
      \IF {$\alpha \eta > {\overline {EE} _{PT}^{\text{samp}}}$}
        \STATE \quad With the given $\eta$, solve \underline{\textbf{P4}} with Algorithm \ref{alg3}. Denote the solutions as ${{\bf{w}}^{\star}}$ and ${\boldsymbol{\phi}}^{\star}$;
      \ELSE
        \STATE \quad Set ${{\bf{w}}^{\star}} = {{\bf{w}}_2^{\star}}$ and ${\boldsymbol{\phi}}^{\star} = {\boldsymbol{\phi}}_2^{\star}$;
      \ENDIF
      \IF {${|{{{\boldsymbol{\phi}}^{\star \mathrm{H}}} \widehat{\mathbf{M}} {\mathbf{w}^{\star}}} |^2} \ge \gamma_{th}$}
        \STATE \quad Set ${\eta_L}=\eta$;
      \ELSE
        \STATE \quad Set ${\eta_U}=\eta$;
      \ENDIF
    \STATE \textbf{Until} ${\frac{{{\eta _U} - {\eta _L}}}{{{\eta _L}}}} <\epsilon $;
    \RETURN ${\eta^{\star}}=\eta$, ${{\bf{w}}^{\star}}$ and ${\boldsymbol{\phi}}^{\star}$.
  \end{algorithmic}
\end{algorithm}
\begin{algorithm}[H]
  \renewcommand{\algorithmicrequire}{\textbf{Input}:}
  \renewcommand{\algorithmicensure}{\textbf{Output}:}
  \renewcommand{\algorithmicif}{\quad \textbf{If}}
  \renewcommand{\algorithmicendif}{\quad \textbf{End if}}
  \renewcommand{\algorithmicelse}{\quad \textbf{Else}}
  \renewcommand{\algorithmicreturn}{\textbf{Return}}
  \caption{The overall algorithm for solving \underline{\textbf{P4}}}
  \label{alg3}
  \begin{algorithmic}[1]
    \STATE \textbf{Initialization}: Feasible values of $\{ {{\bf{w}}^{(0)}} , {{\boldsymbol{\phi}}^{(0)}} \}$ and $i=0$;
    \STATE \textbf{Repeat}
      \STATE \quad Update $i=i+1$;
      \STATE \quad \textbf{Initialization}: Global value of ${{\bf{w}}_g^{(0)}} = {{\bf{w}}^{(i-1)}}$, $k=0$;
      \STATE \quad \textbf{Repeat}
      \STATE \qquad With given ${{\bf{w}}_g^{(k)}}$, solve the convex optimization problem \underline{\textbf{P4-2}}, and denote the optimal solution as ${{\bf{w}}_g^{*\left(k\right)}}$;
      \STATE \qquad Update ${{\bf{w}}_g^{\left(k+1\right)}}={{\bf{w}}_g^{*\left(k\right)}}$;
      \STATE \qquad Update $k=k+1$;
      \STATE \quad \textbf{Until} The fractional increase of the objective value of \underline{\textbf{P4-2}} is below a certain threshold $\kappa_1 $ ;
      \STATE \quad Update ${{\bf{w}}^{(i)}} = {{\bf{w}}_g^{\left(k\right)}}$;
      \STATE \quad \textbf{Initialization}: Global value of ${{\boldsymbol{\phi}}_g^{(0)}} = {{\boldsymbol{\phi}}^{(i-1)}}$, $k=0$;
      \STATE \quad \textbf{Repeat}
      \STATE \qquad With given ${{\boldsymbol{\phi}}_g^{(0)}}$, solve the convex optimization problem \underline{\textbf{P4-6}}, and denote the optimal solution as ${{\boldsymbol{\phi}}_g^{*\left(k\right)}}$;
      \STATE \qquad Update ${{\boldsymbol{\phi}}_g^{\left(k+1\right)}}={{\boldsymbol{\phi}}_g^{*\left(k\right)}}$;
      \STATE \qquad Update $k=k+1$;
      \STATE \quad \textbf{Until} The fractional increase of the objective value of \underline{\textbf{P4-6}} is below a certain threshold $\kappa_2 $ ;
      \STATE \quad Update ${\boldsymbol{\phi}}^{(i)}={{\boldsymbol{\phi}}_g^{\left(k\right)}}$;
    \STATE \textbf{Until} The fractional increase of the objective value of \underline{\textbf{P4}} is below a certain threshold $\kappa_3$;
    \RETURN ${{\bf{w}}^{\star}}={{\bf{w}}^{\left(i\right)}}$ and ${{\boldsymbol{\phi}}^{\star}}={{\boldsymbol{\phi}}^{(i)}}$.
  \end{algorithmic}
\end{algorithm}
As such, the original problem \underline{\textbf{P3}} can be solved based on the bisection method, which is summarized in Algorithm \ref{alg2}, together with Algorithm \ref{alg3}.

\section{Simulation Results}\label{s6}

In this section, simulation results are presented to demonstrate the effectiveness of our proposed algorithm. The main system parameters are listed in Table \ref{table_1} unless otherwise specified. We assume that the PT is equipped with $M=4$ elements, with the adjacent element separated by half wavelength. Moreover, the RIS-BD is distributed in a $N = 8 \times 8 = 64$ rectangular array on the xOz plane. Both the RIS-BD and receiver have equal distance with the PT, which is $d_0 = 300{\text{m}}$, and the angle formed by PT-to-receiver and PT-to-RIS-BD line segments is $\theta$. The height of the PT and the RIS-BD are assumed as $h_{PT} = 50{\text{m}}$ and $h_{RIS-BD} = 30{\text{m}}$. As such, without loss of generality, the coordinate of the PT, receiver and RIS-BD locations can be represented as $\left(0,0,h_{PT}\right)$, $\left(0,d_0,0\right)$ and $\left(d_0 \cos{\theta},d_0 \sin{\theta},h_{RIS-BD}\right) $, respectively.
Therefore, the distance between RIS-BD and receiver projected to the xOy plane can be represented as ${d_1} = 2{d_0} \sin \left( {\frac{{{\theta }}}{2}} \right)$. All links are assumed to be the Rician fading channels, with the Rician K-factor $K_1=2\text{dB}$ in the PT-receiver link and $K_2=10\text{dB}$ in the PT-to-RIS-BD and RIS-BD-to-receiver links.
Moreover, the large-scale path loss is modeled as $\beta = {\beta _0}{d^{ -\alpha }}$, where ${\beta _0} = {\left( {\frac{\lambda }{{4\pi }}} \right)^2}$ is the reference channel gain, $\lambda = {c}/{f_c}$ is the wavelength, $d$ is the distance between the corresponding devices and $\alpha$ denotes the path loss exponent.
  \begin{table}[H]
    \vspace*{-.5\baselineskip}
    \caption{System Parameters}
    \vspace*{-.5\baselineskip}
    \label{table_1}
    \centering
    \small
    \begin{tabular}{cc}
    \hline
    \hline
    \noalign{\smallskip}
    Parameters & Values\\
    \noalign{\smallskip}
    \hline
    \noalign{\smallskip}
    Carrier frequency $f_c$ & $3.5{\text{GHz}}$ \\
    \noalign{\smallskip}
    Channel bandwidth $B$ & $1 \text{MHz}$ \\
    \noalign{\smallskip}
    Path-loss exponent of the PT-receiver link ${\alpha_{TR}}$ & $2.7$ \\
    \noalign{\smallskip}
    Path-loss exponent of the PT-RIS-BD link ${\alpha_{TS}}$ & $2.7$ \\
    \noalign{\smallskip}
    Path-loss exponent of the RIS-BD-receiver link ${\alpha_{SR}}$ & $2.1$ \\
    \noalign{\smallskip}
    Noise power ${\sigma ^2}$ & $-114 \text{dBm}$ \\
    \noalign{\smallskip}
    Reflection coefficient $\rho$ & $1$ \\
    \noalign{\smallskip}
    Inefficiency of the power amplifier of the PT $\mu$ & $1.2$ \\
    \noalign{\smallskip}
    Circuit power of the PT $P_s$ & $39 \text{dBm}$ \\
    \noalign{\smallskip}
    Circuit power of each RIS-BD element $P_r$ & $10 \text{dBm}$ \\
    \noalign{\smallskip}
    Ratio between symbol duration of the RIS-BD to the PT $L$ & $128$ \\
    \noalign{\smallskip}
    \hline
    \end{tabular}
    \vspace*{-.5\baselineskip}
\end{table} 

For $P_{\max} = 40 \text{dBm}$, Fig. \ref{fig_2} plots the outer boundaries of the EE regions obtained by our proposed sample-average based approach, which are labelled as the ``EE-max design". Note that each point of these boundaries corresponds to an EE pair by varying the EE profile as analyzed in Subsection \ref{s5-1}. It is observed from Fig. \ref{fig_2} that there exists a non-trivial EE trade-off between PT and RIS-BD, i.e., a sacrifice of the EE for the PT would lead to considerable improvement to that of the RIS-BD, and vice versa. Moreover, these EE regions exhibit convex characteristic, which is different from the observation in the single-antenna BD based PSR case \cite{10039158}.
We also consider the so-called ``PT-rate-max design" as a benchmark comparison, where the transmit beamforming is designed to maximize the primary communication rate, rather than the EE. Note that the ``RIS-BD-rate-max design" is equivalent to maximizing the individual EE of the RIS-BD as considered in Subsection \ref{s3-2}, which is omitted in the figures. Therefore, another observation from Fig. \ref{fig_2} is that the resulting EE pairs by the ``PT-rate-max design" lie in the interior of the achievable EE region. This implies that simply maximizing the communication rate is strictly energy-inefficient, which demonstrates the importance of considering EE metrics deliberately in SR systems.

By comparing the different curves under different $\theta$ in Fig. \ref{fig_2:a}, we can see that when $\theta$ is relatively small, i.e., $\theta \le 15^\circ$, which means that the RIS-BD is very close to the receiver, the achievable EE region enlarges as $\theta$ decreases. While when $\theta > 15^\circ$, where the RIS-BD is relatively far from the receiver, instead of completely decreasing with the increase of $\theta$, the EE of the PT and RIS-BD show the opposite trend as $\theta$ goes larger. This is expected since a relatively small $\theta$ value implies that the PT-to-receiver channel $\bf{h}$ and PT-to-RIS-BD channel $\bf{G}$ are highly correlated, where the mutual benefit of SR communication systems is more significant. However, a larger $\theta$ value implies that $\bf{h}$ and $\bf{G}$ are less correlated, which makes it more challenging to find the optimal transmit beamforming so that the significant power can be simultaneously directed towards both PT and RIS-BD, where the EE trade-off between PT and RIS-BD is more prominent.
Fig. \ref{fig_2:b} plots the outer boundaries of achievable EE regions and the resulting EE pairs with the ``PT-rate-max design" for two different circuit power levels $P_s = 0$ and $39\text{dBm}$, with $P_{\max} = 40 \text{dBm}$ and $\theta = 20^\circ $. It is observed from Fig. \ref{fig_2:b} that as the circuit power consumption reduces, the achievable EE region enlarges, as expected. Moreover, unlike the EE regions for $P_s = 39\text{dBm}$ in Fig. \ref{fig_2:a}, which are convex, the achievable EE region for the extreme case $P_s = 0$ is concave in comparison. This is expected due to the fact that the power consumed by the PT plunges sharply with $P_s$ going to zero, which contributes to the surge of the EE of the PT, while has no effect on the EE of RIS-BD.
\vspace*{-.5\baselineskip}
\begin{figure}[H]
  \centering
  \subfigure[with different $\theta$]{
  \label{fig_2:a}
  \includegraphics[width = .46\textwidth]{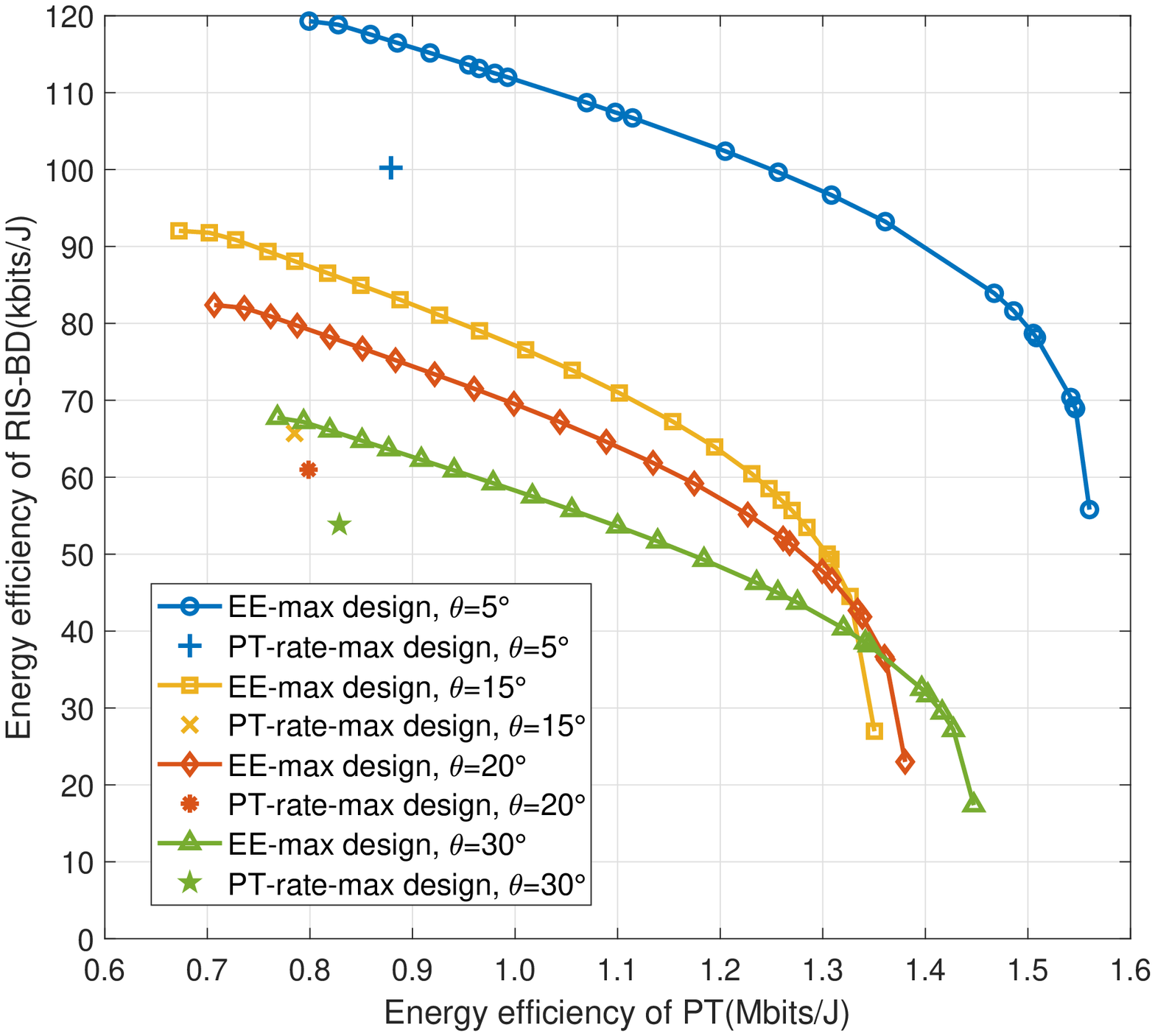}}
  \subfigure[with $P_s = 0$ and $P_s = 39\text{dBm}$]{
  \label{fig_2:b}
  \includegraphics[width = .51\textwidth]{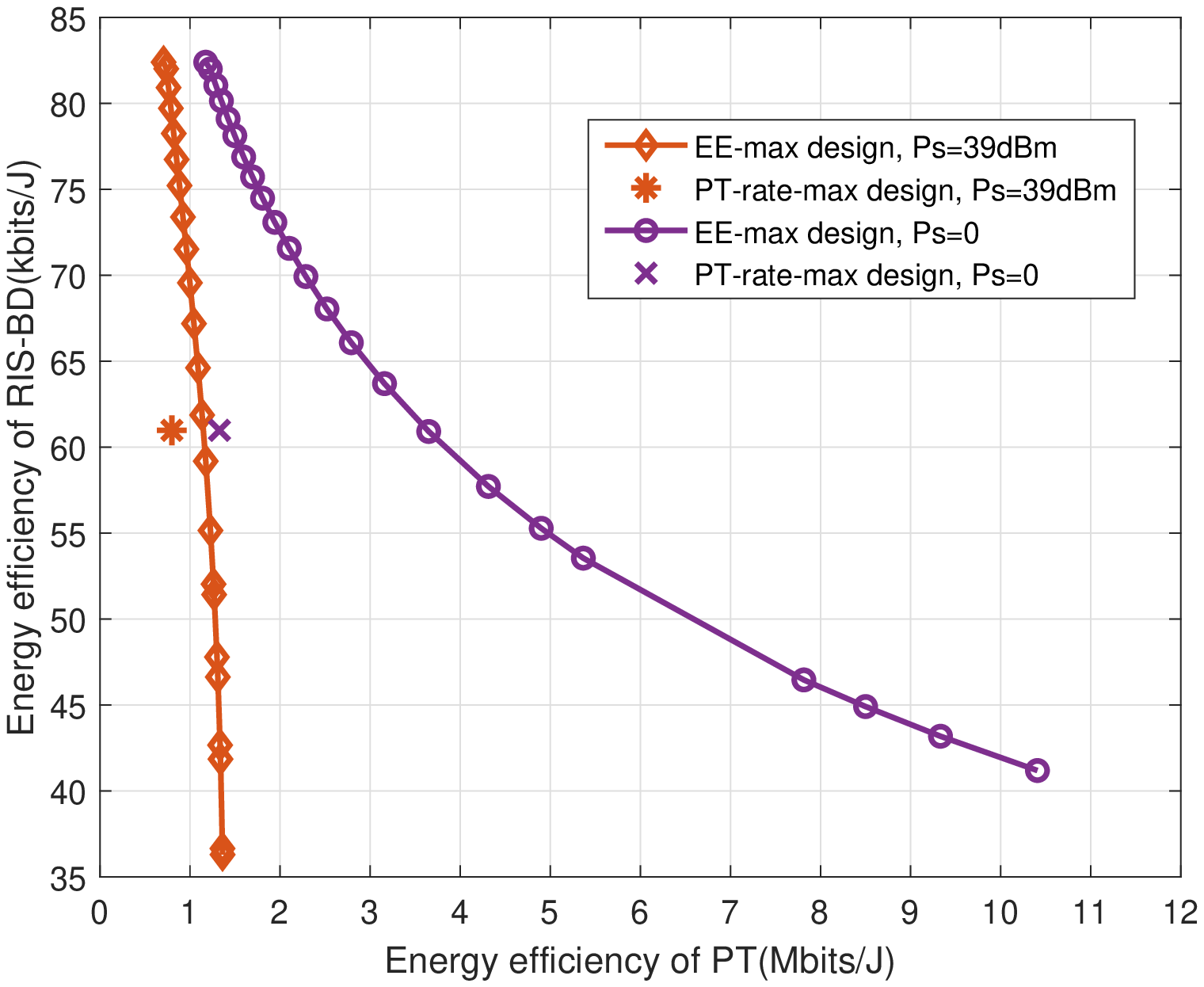}}
  \vspace*{-.5\baselineskip}
  \caption{EE region of the RIS-BD-based SR system.}
  \vspace*{-.5\baselineskip}
  \label{fig_2}
\end{figure}

Fig. \ref{fig_10} plots the outer boundaries of achievable EE regions for five different maximum transmit power $P_{\max}$ with $P_s = 39\text{dBm}$ and $\theta = 20^\circ$.
It is observed from the two subfigures in Fig. \ref{fig_10} that the achievable EE region does not enlarge indefinitely with the increase of $P_{\max}$. For example, when $P_{\max} < 38 \text{dBm}$ in Fig. \ref{fig_10:a}, the achievable EE region enlarges as $P_{\max}$ increases, while when $P_{\max} \ge 38 \text{dBm}$ in Fig. \ref{fig_10:b}, there is little change in the achievable EE region even if $P_{\max}$ increases. This is expected since the definitions of EEs of the PT and RIS-BD as in (\ref{eq15}) and (\ref{eq25}) imply that increasing the transmit power do not lead to monotonically increasing of EE.
\vspace*{-.5\baselineskip}
\begin{figure}[H]
  \centering
  \subfigure[$P_{\max}$ is relatively small]{
  \label{fig_10:a}
  \includegraphics[width = .46\textwidth]{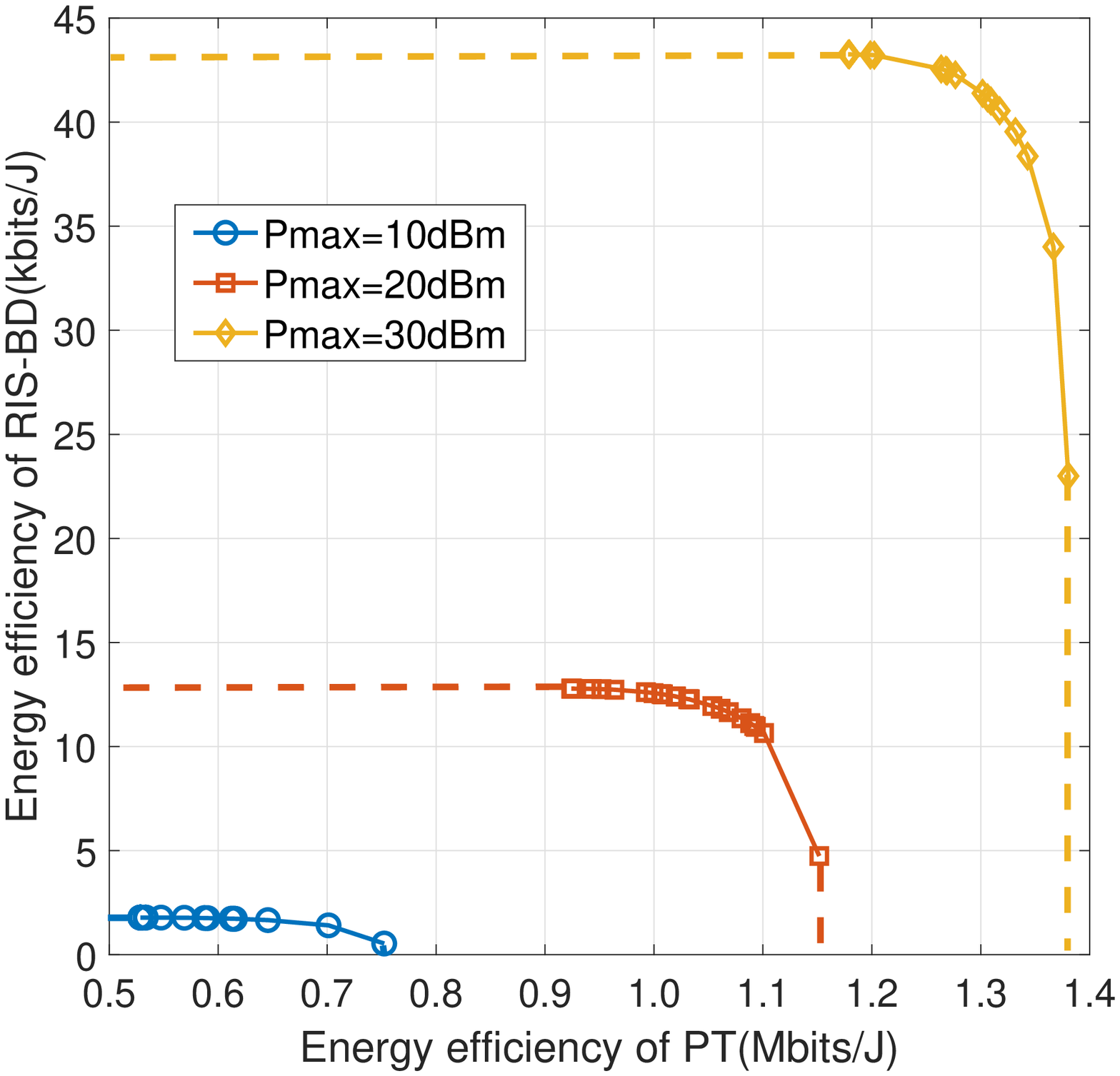}}
  \subfigure[$P_{\max}$ is large]{
  \label{fig_10:b}
  \includegraphics[width = .49\textwidth]{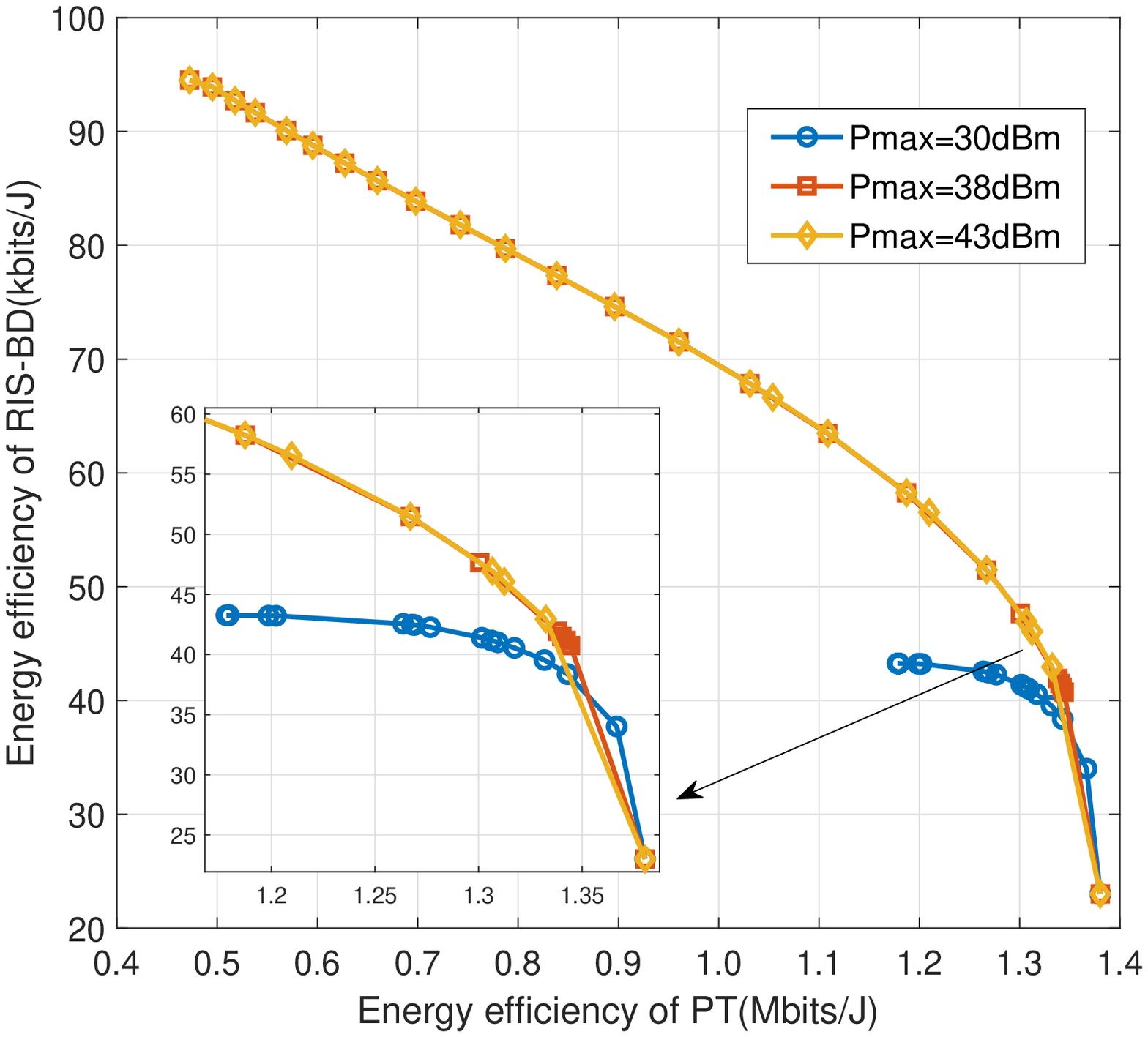}}
  \vspace*{-.5\baselineskip}
  \caption{EE region of the RIS-BD-based SR system with different $P_{\max}$.}
  \vspace*{-.5\baselineskip}
  \label{fig_10}
\end{figure}

\section{Conclusion}\label{s7}
This paper studied the EE trade-off of the active and passive communications. The maximum individual EE of the PT and RIS-BD, and the ir asymptotic closed-forms are derived, which reveal that there exist non-trivial trade-off between these two EEs. By applying the sample-average based bisection approach together with AO algorithm and SCA technique, an optimization problem is formulated and efficiently solved to characterize the Pareto boundary of the EE region. Finally, simulation results have validated our theoretical analysis and demonstrated the effectiveness of the proposed algorithms.


%
\appendices
\section{Proof of Proposition \ref{Pro1}}\label{App3}
With ${\bf{w}}=\sqrt p {\bf{v}}$ and ${\mathbf{F}} = \widehat {\mathbf{h}}{\widehat {\mathbf{h}}^{\mathrm{H}}} + {\widehat {\mathbf{M}}^{\mathrm{H}}}\phi {\phi ^{\mathrm{H}}}\widehat {\mathbf{M}}$, we can rewrite the objective function in \underline{\textbf{P1-1}} as
\begin{equation}\label{eq51}
  {EE}_{PT}^{{\mathrm{ub}}} = \frac{{B{{\log }_2}(1 + p{{\mathbf{v}}^{\mathrm{H}}}{\mathbf{Fv}})}}{{\mu p + {P_s}}}.
\end{equation}
Since the optimum solution can be attained at each iteration, we have
\begin{equation}\label{eq56}
  \begin{aligned}
    {EE}_{PT}^{{\mathrm{ub}}( {i + 1} )} & = \frac{{B{{\log }_2}(1 + {p^{\left( {i + 1} \right)}}{{\mathbf{v}}^{\left( {i + 1} \right)}}^{\mathrm{H}}{{\mathbf{F}}^{\left( {i + 1} \right)}}{{\mathbf{v}}^{\left( {i + 1} \right)}})}}{{\mu {p^{\left( {i + 1} \right)}} + {P_s}}} \ge \frac{{B{{\log }_2}(1 + {p^{\left( {i + 1} \right)}}{{\mathbf{v}}^{\left( {i + 1} \right)}}^{\mathrm{H}}{{\mathbf{F}}^{\left( i \right)}}{{\mathbf{v}}^{\left( {i + 1} \right)}})}}{{\mu {p^{\left( {i + 1} \right)}} + {P_s}}} \\
    & \ge \frac{{B{{\log }_2}(1 + {p^{\left( i \right)}}{{\mathbf{v}}^{\left( {i + 1} \right)}}^{\mathrm{H}}{{\mathbf{F}}^{\left( i \right)}}{{\mathbf{v}}^{\left( {i + 1} \right)}})}}{{\mu {p^{\left( i \right)}} + {P_s}}} \ge \frac{{B{{\log }_2}(1 + {p^{\left( i \right)}}{{\mathbf{v}}^{\left( i \right)}}^{\mathrm{H}}{{\mathbf{F}}^{\left( i \right)}}{{\mathbf{v}}^{\left( i \right)}})}}{{\mu {p^{\left( i \right)}} + {P_s}}} = EE{_{PT}^{{\mathrm{ub}}\left( i \right)}},
  \end{aligned}
\end{equation}
where $i$ is the iteration index in Algorithm \ref{alg1}. Thus, the relationship (\ref{eq56}) shows that the objective value of \underline{\textbf{P1-1}} obtained in Algorithm \ref{alg1} is monotonically non-decreasing after each iteration, and hence converges to a finite limit.
This completes the proof of Proposition \ref{Pro1}.

\section{Proof of Lemma \ref{Lem1}}\label{App4}
Based on (\ref{eq37}), we have
\begin{equation}\label{eq60}
  \begin{aligned}
    & {\mathbf{D}} = {\mathbf{h}}{{\mathbf{h}}^{\mathrm{H}}} + \rho \sum\nolimits_{n = 1}^N {{{\left| {{f_n}} \right|}^2}{\mathbf{g}}_n^* {\mathbf{g}}_n^\mathrm{T}}  + \rho \sum\nolimits_{n \ne n'}^N {\sum\nolimits_{n' = 1}^N {{f_n}f_{n'}^* {e^{j\left( {{\theta _{n'}} - {\theta _n}} \right)}}{\mathbf{g}}_n^* {\mathbf{g}}_{n'}^\mathrm{T}} }  \\
    & \mathop  \to \limits^{(a)} M\big( {{\beta _{TR}} + \rho \sum\nolimits_{n = 1}^N {{{\left| {{f_n}} \right|}^2}{\beta _{TS}}} } \big){{\mathbf{I}}_M}  \mathop  \to \limits^{(b)} M\left( {{\beta _{TR}} + \rho N{\beta _{SR}}{\beta _{TS}}} \right){{\mathbf{I}}_M},
  \end{aligned}
\end{equation}
where $(a)$ results from the law of large numbers and the assumption of i.i.d. Rayleigh fading channels that for $M \gg 1$, we have $\frac{1}{M}{\mathbf{h}}{{\mathbf{h}}^{\mathrm{H}}} \to {\beta _{TR}}{{\mathbf{I}}_M}$, $\frac{1}{M}{\mathbf{g}}_n^* {\mathbf{g}}_n^\mathrm{T} \to {\beta _{TS}}{{\mathbf{I}}_M}$ and $\frac{1}{M}{\mathbf{g}}_n^* {\mathbf{g}}_{n'}^\mathrm{T} \to {{\mathbf{O}}_M}$ for $n \ne n'$, and $(b)$ holds since $\frac{1}{N}\sum\nolimits_{n = 1}^N {{| {{f_n}}|^2}} \to \mathbb{E}\left[ {{|{{f_n}}|^2}} \right] = {\beta _{SR}}$ for $N \gg 1$. 
Therefore, $EE_{PT}^{\mathrm{ub}}$ is now independent of the phase shift of the RIS-BD, and it reduces to
\begin{equation}\label{eq38}
  EE_{PT}^{\mathrm{ub}} \to\frac{{B{{\log }_2}\Big(1 + \frac{p}{{{\sigma ^2}}}M \left( {{\beta _{TR}} + \rho N{\beta _{SR}}{\beta _{TS}}} \right) {{\mathbf{v}}^{\mathrm{H}} {\mathbf{v}}}\Big)}}{{\mu p + {P_s}}}.
\end{equation}
Due to the fact that ${\mathbf{v}}^{\mathrm{H}} {\mathbf{v}} = \left\| {\mathbf{v}} \right\|^2 = 1$, the proof is thus completed.

\section{Proof of Lemma \ref{Lem2}}\label{App1}
Firstly, we have the following channel assumptions:
\begin{equation}\label{eq68}
  \begin{aligned}
    {f_n} & \sim \mathcal{CN}\Big( {\sqrt {{\frac{{\beta _{SR}}{{K_3}}}{{{K_3} + 1}}}} {f_{LoS,n}},\frac{{{\beta _{SR}}}}{{ {{K_3} + 1} }}\sum\nolimits_{r = 1}^N {|{{\left( {{{\mathbf{R}}_{SR}}} \right)}_{n,r}}|} } \Big), \\
    {g_n} & \sim \mathcal{CN}\Big( {\sqrt {{\frac{{\beta _{TS}}{{K_2}}}{{{K_2} + 1}}}} {g_{LoS,n}},\frac{{{\beta _{TS}}}}{{ {{K_2} + 1} }}\sum\nolimits_{r = 1}^N {|{{\left( {{{\mathbf{R}}_{TS}}} \right)}_{n,r}}|} } \Big).
  \end{aligned}
\end{equation}
Then according to \cite{9357969}, the means of $|f_n|$ and $|g_n|$ are
\begin{equation}\label{eq69}
  \begin{aligned}
    {\mu _{f,n}} = & \sqrt {\frac{{\beta _{SR}} \pi {\sum\nolimits_{r = 1}^N {\left| {{{({{\mathbf{R}}_{SR}})}_{n,r}}} \right|}} }{{4({K_3} + 1)}}}  \times {\mathrm{L}_{\frac{1}{2}}}\Big( { - \frac{{{K_3}}}{{\sum\nolimits_{r = 1}^N | {{({{\mathbf{R}}_{SR}})}_{n,r}}|}}} \Big), \\
    {\mu _{g,n}} = & \sqrt {\frac{{\beta _{TS}} \pi {\sum\nolimits_{r = 1}^N | {{({{\mathbf{R}}_{TS}})}_{n,r}}|}}{{4({K_2} + 1)}}}  \times {\mathrm{L}_{\frac{1}{2}}}\Big( { - \frac{{{K_2}}}{{\sum\nolimits_{r = 1}^N | {{({{\mathbf{R}}_{TS}})}_{n,r}}|}}} \Big).
  \end{aligned}
\end{equation}
According to the property of the product of independent random variables, we have $\mathbb{E} \left[ |{f_n}||{g_n}| \right] = {\mu _{f,n}}{\mu _{g,n}}$.
Based on the central limit theorem (CLT), it can be shown that $X = \sum\nolimits_{n = 1}^N {|{f_n}||{g_n}|}$ is the sum of $N$ independently distributed random variables, which follows the Gaussian distribution for $N \gg 1$. Therefore, the mean is $\mathbb{E} \left[ X \right] = \sum\nolimits_{n = 1}^N {{\mu _{f,n}}{\mu _{g,n}}}$, and the proof is thus completed.

\section{Proof of Lemma \ref{Lem3}}\label{App2}
It is not difficult to see from (\ref{eq14}) that the random variable $Y$ is the sum of $M$ random variables $\left|\sum\nolimits_{n=1}^N {{f_n^* } {g_{n m}} e^{j\theta_n}}\right|^2$. Take $\forall m$ for an example, we define $T_m = {\sum\nolimits_{n = 1}^N {f_n^* {g_{n{m}}}{e^{j{\theta _n}}}} }$.
With similar channel assumptions as (\ref{eq68}), and defining ${\mu _{f,n}} = \sqrt {\frac{{{\beta _{SR}}{K_3}}}{{{K_3} + 1}}} {f_{LoS,n}}, {\mu _{g,n{m}}} = \sqrt {\frac{{{\beta _{TS}}{K_2}}}{{{K_2} + 1}}} {g_{LoS,n{m}}}$, and $\sigma _{f,n}^2 = \frac{{{\beta _{SR}}}}{{{K_3} + 1}}\sum\nolimits_{r = 1}^N {|{{\left( {{{\mathbf{R}}_{SR}}} \right)}_{n,r}}|}$, $\sigma _{g,n}^2 = \frac{{{\beta _{TS}}}}{{{K_2} + 1}}\sum\nolimits_{r = 1}^N {|{{\left( {{{\mathbf{R}}_{TS}}} \right)}_{n,r}}|} $ for convenience, we have
\begin{equation}\label{eq80}
  \mathbb{E}\left[ {f_n^* {g_{nm}}} \right] = {\mu _{n,m}} = \mu _{f,n}^* {\mu _{g,nm}}, \;
  \mathrm{Var}\left[ {f_n^* {g_{nm}}} \right] = \sigma _{n,m}^2 = \sigma _{f,n}^2\sigma _{g,n}^2 + \mu _{f,n}^2\sigma _{f,n}^2 + \mu _{g,nm}^2\sigma _{g,n}^2.
  \end{equation}
Based on the CLT, it can be shown that $T_m$, which is the sum of $N$ independently distributed random variables ${f_n^* {g_{nm}}}$, follows the complex Gaussian distribution for $N \gg 1$ with the following mean and variance: $\mu_m = \sum\nolimits_{n = 1}^N {{\mu _{n,m}}{e^{j{\theta _n}}}}$ and $\sigma _m^2 = \sum\nolimits_{n = 1}^N {\sigma _{n,m}^2}$. Then it can be decomposed into real part and imaginary part, which are both independent Gaussian random variables with the same variance $\frac{1}{2}\sigma_m^2 $. Therefore, $T_m^2 = {\operatorname{Re} {{\left\{ {{T_m}} \right\}}^2} + \operatorname{Im} {{\left\{ {{T_m}} \right\}}^2}}$ is distributed as a non-central chi-square distribution with $2$ degrees of freedom, whose non-centrality parameter is $\lambda_m  = \mu_m^2$. Then, the mean of $T_m^2$ is $\mathbb{E}[T_m^2]={\sigma_m ^2} + \lambda_m$. According to the additivity of chi-square distribution, we have $Y = \sum\nolimits_{m = 1}^M {T_m^2}$ is distributed as a non-central chi-square distribution with $2M$ degrees of freedom, whose non-centrality parameter is $\lambda  = \sum\nolimits_{m = 1}^M {\lambda_m}$. Therefore, the mean of $Y$ is $\mathbb{E}\left[ Y \right] = \sum\nolimits_{m = 1}^M {\mathbb{E}[T_m^2]} = \sum\nolimits_{m = 1}^M {{\sigma_m ^2} + \lambda_m} $.
This completes the proof of Lemma \ref{Lem3}.

  \bibliographystyle{IEEEtran} 
  \bibliography{reference}

\end{document}